%% file: geodesic_spanners.tex
\documentclass[a4paper,UKenglish,cleveref, autoref, thm-restate]{lipics-v2021}

\usepackage{microtype} % if unwanted, comment out
\hideLIPIcs 
\title{The Complexity of Geodesic Spanners}
\author{Sarita de Berg}{Department of Information and Computing Sciences, Utrecht University, The Netherlands}{S.deBerg@uu.nl}{}{}%TODO mandatory, please use full name; only 1 author per \author macro; first two parameters are mandatory, other parameters can be empty. Please provide at least the name of the affiliation and the country. The full address is optional

\author{Marc van Kreveld}{Department of Information and Computing Sciences, Utrecht University, The Netherlands}{M.J.vanKreveld@uu.nl}{}{}

\author{Frank Staals}{Department of Information and Computing Sciences, Utrecht University, The Netherlands}{F.Staals@uu.nl}{}{}

\authorrunning{S. de Berg, M. van Kreveld, F. Staals}

\Copyright{Sarita de Berg, Marc van Kreveld, and Frank Staals}

 %TODO mandatory: Please choose ACM 2012 classifications from https://dl.acm.org/ccs/ccs_flat.cfm 

 %TODO mandatory; please add comma-separated list of keywords

\category{} %optional, e.g. invited paper

\relatedversion{A preliminary version of this paper appeared at the International Symposium on Computational Geometry (SoCG) 2023~\cite{complexity_spanners}.} %optional, e.g. full version hosted on arXiv, HAL, or other respository/website
\nolinenumbers %uncomment to disable line numbering

\usepackage{graphicx}
\usepackage{amsfonts}
\usepackage{xspace}

\usepackage{caption}
\usepackage[title]{appendix}
\graphicspath{{figures/}}
\usepackage{pifont}
\usepackage{color}

\graphicspath{ {./figures/} }

\newcommand {\mathset} [1] {\ensuremath {\mathbb {#1}}\xspace}
\newcommand {\R} {\mathset {R}}

\newcommand{\mkmcal}[1]{\ensuremath{\mathcal{#1}}\xspace}

\newcommand{\G}{\mkmcal{G}}

\newcommand{\T}{\mkmcal{T}}
\renewcommand{\P}{P}
\newcommand{\VD}{\ensuremath{\mkmcal{VD}}\xspace}
\newcommand{\HD}{\ensuremath{\mkmcal{HD}}\xspace}

\newcommand{\Int}{\text{Int}}
\newcommand{\ilambda}{\ensuremath{^{(i)}_\lambda}}

\newcommand{\thmheadfont}{\textcolor{lipicsGray}{$\blacktriangleright$}\nobreakspace\bfseries\sffamily}
\newenvironment{repeatenv}[2]%
  {\smallskip\noindent {\thmheadfont #1~\ref{#2}.}\ \slshape}
  {\normalfont}

\newcommand{\eps}{\ensuremath{\varepsilon}\xspace}
\newcommand{\etal}{et al.\xspace}

\def\polylog{\operatorname{polylog}}

\begin{document}

\maketitle

\begin{abstract}
A \emph{geometric $t$-spanner} for a set $S$ of $n$ point sites is an edge-weighted graph for which the (weighted) distance between any two sites $p,q \in S$ is at most $t$ times the original distance between $p$ and~$q$. We study geometric $t$-spanners for point sets in a constrained two-dimensional environment $P$. In such cases, the edges of the spanner may have non-constant complexity. Hence, we introduce a novel spanner property: the spanner \emph{complexity}, that is, the total complexity of all edges in the spanner. Let $S$ be a set of $n$ point sites in a simple polygon $P$ with $m$ vertices. We present an algorithm to construct, for any fixed integer $k \geq 1$, a $2\sqrt{2}k$-spanner with complexity $O(mn^{1/k} + n\log^2 n)$ in $O(n\log^2n + m\log n + K)$ time, where $K$ denotes the output complexity. When we relax the restriction that the edges in the spanner are shortest paths, such that an edge in the spanner can be any path between two sites, we obtain for any constant $\varepsilon \in (0,2k)$ a \emph{relaxed} geodesic $(2k + \eps)$-spanner of the same complexity, where the constant is dependent on $\varepsilon$. When we consider sites in a polygonal domain $P$ with holes, we can construct a relaxed geodesic $6k$-spanner of complexity $O(mn^{1/k} + n\log^2 n)$ in $O((n+m)\log^2n\log m+ K)$ time. Additionally, for any constant $\varepsilon \in (0,1)$ and integer constant $t \geq 2$, we show a lower bound for the complexity of any $(t-\varepsilon)$-spanner of $\Omega(mn^{1/(t-1)} + n)$.
\end{abstract}

%%%%%%%%%%%%%%%%%%%%%%%%%%%%%%%%%%%%%%%%%%%%%%%%%%%%%%%%%%%%%%%%%%%
\section{Introduction}
In the design of networks on a set of nodes, we often consider two criteria: few connections between the nodes, and small distances. Spanners are geometric networks on point sites that replace the small distance criterion by a small detour criterion.
Formally, a \emph{geometric $t$-spanner} for a set $S$ of $n$ point sites is an edge-weighted graph $\G = (S,E)$ for which the (weighted) distance $d_\G(p,q)$ between any two sites $p,q \in S$ is at most $t \cdot d(p,q)$, where $d(p,q)$ denotes the distance between $p$ and $q$ in the distance metric we consider~\cite{proximity_algorithms_book}. 
The smallest value $t$ for which a graph $\G$ is a $t$-spanner is called the \emph{spanning ratio} of $\G$. The number of edges in the spanner is called the \emph{size} of the spanner.

In the real world, spanners are often constructed in some sort of environment. For example, we might want to connect cities by a railway network, where the tracks should avoid obstacles such as mountains or lakes. One way to model such an environment is by a polygonal domain.
In this paper, we study the case where the sites lie in a polygonal domain~$P$ with $m$ vertices and $h$ holes,
and we measure the distance between two points $p,q$ by their \emph{geodesic
distance}: the length of the shortest path between $p$ and $q$ fully
contained within $P$. An example of such a spanner is provided in Figure~\ref{fig:geodesic_spanner}.

\begin{figure}
    \centering
    \includegraphics{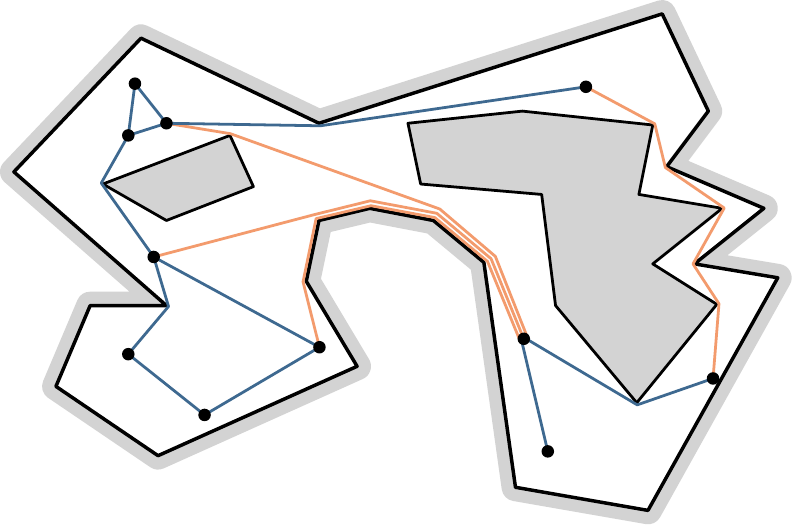}
    \caption{A spanner on a set of point sites in a polygonal domain. Because of the orange edges, the spanner has a relatively high complexity.}
    \label{fig:geodesic_spanner}
\end{figure}

The spanning ratio and the size of spanners are not the only properties of
interest. Many different properties have been
studied, such as total weight (or lightness), maximum degree, (hop)
diameter, and
fault-tolerance~\cite{Arya95short_thin_lanky,Bose05Bounded_degree_low_weight, survey_geometric_spanners, Chan15Doubling_spanners,Elkin15optimal_short_thin_lanky,Le19truly_opt_Euclidean,Levcopoulos98_fault_tolerant,book_spanners}. When
we consider distance metrics for which the edges in the spanner no
longer have constant complexity, another interesting property of
spanners arises: the spanner \emph{complexity}, i.e.\ the total
complexity of all edges in the spanner. In our railway example, this corresponds to the total number of bends in the tracks. A spanner with a low number of bends may be desired, as trains can drive faster on straight tracks, and it makes construction cheaper.
In this paper, we study this
novel property for point sites in a polygonal domain, where the complexity
of an edge is simply the number of line segments in the path. In this
setting, a single edge may have complexity $\Theta(m)$. Naively, a
spanner of size $E$ could thus have complexity $\Theta(mE)$. Our goal
is to compute an $O(1)$-spanner of size $O(n \polylog n)$ with small
complexity, preferably near linear in both $n$ and $m$.

When studying spanning trees of points, two variants exist: with or
without Steiner points. The same is true for spanners, where Steiner points can be used to obtain lighter and sparser spanners~\cite{Bhore22Steiner_spanners_light_sparse,Le19truly_opt_Euclidean}. 
In this paper we focus on the variant where Steiner points
are \emph{not} allowed, leaving the other variant to future research.

\subparagraph{Related work.} For the Euclidean distance in $\R^d$, and
any fixed $\varepsilon > 0$, there is a $(1 + \varepsilon)$-spanner of
size $O(n/\eps^{d-1})$~\cite{book_spanners}. For the more general case
of metric spaces of bounded doubling dimension we can also construct a
$(1 + \varepsilon)$-spanner of size
$O(n/\eps^{O(d)})$~\cite{RoutingInDoublingMetrics,bounded_doubling_optimal_dynamic,FastConstructionDoublingMetrics}.
These results do not apply when the sites lie in a polygon,
and we measure their distances using the geodesic distance.  Abam
\etal~\cite{SpannerPolygonalDomain} show there is a set of $n$ sites
in a simple polygon $P$ for which any geodesic
$(2-\varepsilon)$-spanner has $\Omega(n^2)$ edges. They also construct
a geodesic $(\sqrt{10} +\eps)$-spanner of size $O(n \log^2 n)$ for
sites in a simple polygon, and a geodesic $(5 + \eps)$-spanner of size
$O(n \sqrt{h} \log^2n)$ for sites in a polygonal domain.  Recently,
Abam, de Berg, and Seraji~\cite{SpannerPolyhedralTerrain} showed that a geodesic
$(2+\varepsilon)$-spanner with $O(n\log n)$ edges exists for points on
a polyhedral terrain, thereby almost closing the gap between the upper
and lower bound on the spanning ratio. However, they show only the
existence of such a spanner, and leave constructing one open. Moreover, all of these spanners can have high, $\Omega(nm)$,
complexity.

Abam, de Berg, and Seraji~\cite{SpannerPolyhedralTerrain} make
use of spanners on an \emph{additively weighted} point set in~$\R^d$. In this
setting, the distance between two sites $p,q$ is $w(p) + |pq| + w(q)$
for $p \neq q$, where $w(p)$ is the non-negative weight of a site
$p \in S$ and $|pq|$ denotes the Euclidean distance, and $0$ for
$p=q$. Such additively weighted spanners were studied before by Abam
\etal~\cite{SpannerAdditivelyWeighted}, who obtain a
$(5 + \eps)$-spanner of linear size, and a $(2+\eps)$-spanner of
size $O(n\log n)$. They also provide a lower bound of
$\Omega(n^2)$ on the size of any $(2-\eps)$-spanner. Abam
\etal~\cite{SpannerPolyhedralTerrain} improve these results and obtain
a nearly optimal additively weighted $(2 + \eps)$-spanner of size~$O(n)$.

The other key ingredient for the geodesic $(2+\eps)$-spanner of Abam
\etal~\cite{SpannerPolyhedralTerrain} is a \emph{balanced
  shortest-path separator}. Such a separator consists of either a
single shortest path between two points on the boundary of the
terrain, or three shortest paths that form a \emph{shortest-path
  triangle}. This separator partitions the terrain into two
subterrains, and we call it balanced when each of these terrains
contains roughly half of the sites in $S$. 
In their constructive proof
for the existence of such a balanced separator, they assume that the
three shortest paths in a shortest-path triangle are disjoint, except
for their mutual endpoints. However, during their construction it can
actually happen that these paths are \emph{not} disjoint. When this
happens, it is unclear exactly how to proceed. In Appendix~\ref{ap:separator}, we correct this technical issue. Just like for the
$(2+\eps)$-spanner, the computation of a balanced separator is left
for future research. 
We show how to compute such a separator efficiently in a polygonal domain,
using techniques proposed by Thorup~\cite{Thorup_separator_polygon}, and Lipton and Tarjan~\cite{separatorTheorem}.

Next to spanners on the complete Euclidean geometric graph, spanners under line segment constraints were studied~\cite{Bose19_vis_constraint,Bose_vis_constraint,Bose19_vis_yao,Clarkson87_approx_sp, Das97_vis_constraint}. In this setting, a set $C$ of line segment constraints is provided, where each line segment is between two sites in $S$ and no two line segments properly intersect. The goal is to construct a spanner on the \emph{visibility graph} of $S$ with respect to $C$. Clarkson~\cite{Clarkson87_approx_sp} showed how to construct a linear sized $(1+\eps)$-spanner for this graph. Later, (constrained) Yao- and $\Theta$-graphs were also considered in this setting~\cite{Bose19_vis_constraint,Bose19_vis_yao}. If the segments in $C$ form a polygonal domain $P$, this setting is similar to ours, except that \emph{all} vertices of $P$ are included as sites in $S$. Thus the complexity of each edge is constant, and additionally it is required that there are short paths between the vertices of $P$.

Low complexity paths are studied in the \emph{minimum-link path} problem. In this problem, the goal is to find a path the uses a minimal number of links (edges) between two sites in a domain, for example a simple polygon~\cite{Ghosh_min_link,HomotopicPaths,Mitchell_min_link_paths,suri_linear_min_link}. Generally, this problem focuses only on the complexity of the path, with no restriction on the length of the path. Mitchell~\etal~\cite{Mitchell_shortest_k_link} consider the related problem of finding the \emph{shortest} path with at most $k$ edges between two points $p,q$ in a simple polygon. They give an algorithm to compute a $k$-link path with length at most $(1+\eps)$ times the length of the shortest $k$-link path, for any $\eps >0$. This result can not be applied to our setting, as the length of our paths should be bounded in terms of $d(p,q)$, i.e. the shortest $m$-link path, instead of the shortest $k$-link path.

\subparagraph{Our results.} 
We first consider the simple setting where the sites lie in a simple polygon, i.e.\ a polygonal domain without holes. 
We show that in this setting any $(3-\varepsilon)$-spanner may have complexity $\Omega(nm)$, thus implying that the $(2 + \varepsilon)$-spanner of Abam, de Berg, and Seraji~\cite{SpannerPolyhedralTerrain} may also have complexity $\Omega(nm)$, despite having $O(n\log n)$ edges.

To improve this complexity, we first introduce a simple 2-spanner with
$O(n\log n)$ edges for an additively weighted point set in a
1-dimensional Euclidean space; see Section~\ref{sec:1D-spanner}.  In
Section~\ref{sec:simple_polygon}, we use this result to obtain a
geodesic $2\sqrt{2}$-spanner with $O(n\log^2n)$ edges for a point set
in a simple polygon. We recursively split the polygon by a chord $\lambda$ such that each subpolygon contains roughly half of the sites, and build a 1-dimensional spanner on the sites projected to $\lambda$. We then extend this spanner into one
that also has bounded complexity. For
any fixed integer $k \geq 1$, we obtain a
$2\sqrt{2}k$-spanner with complexity
$O(mn^{1/k} + n\log^2 n)$.
Furthermore, we provide an algorithm to compute
such a spanner that runs in $O(n\log^2 n + m \log n + K)$ time, where
$K$ denotes the output complexity. When we output each edge
explicitly, $K$ is equal to the spanner complexity. However, as each
edge is a shortest path, we can also output an edge implicitly by
only stating the two sites it connects. In this case $K$ is equal to
the size of the spanner.

In Sections~\ref{sub:sp-separator} and~\ref{sec:polygonal_domain}, we extend our results for a
simple polygon to a polygonal domain. There are two significant difficulties
in this transition: \textit{(i)} we can no longer partition the polygon
by a line segment such that each subpolygon contains roughly half of
the sites, and \textit{(ii)} the shortest path between two sites $p,q$
may not be homotopic to the path from $p$ to $q$ via
another site $c$.

We solve problem \textit{(i)} by using a shortest-path
separator similar to Abam, de Berg, and Seraji~\cite{SpannerPolyhedralTerrain}.
To apply the shortest-path separator in a polygonal domain, we need new additional ideas, which we discuss in Section~\ref{sub:sp-separator}.
In particular, we allow one additional type
of separator in our version of a shortest-path separator: two shortest paths from a point in $P$ to the boundary
of a single hole. We show that this way there indeed always exists such a
separator in a polygonal domain, and provide an
$O(m \log m + n \log n)$ time algorithm to compute one, using ideas of Thorup~\cite{Thorup_separator_polygon}. 

To overcome problem \textit{(ii)}, we allow an edge $(p,q)$ to be any path from $p$ to
$q$. In networks, the connections between two nodes are often not
necessarily optimal paths, the only requirement being that the distance between two
hubs does not become too large. Thus allowing other paths between two
sites seems a reasonable relaxation. We call such a spanner a \emph{relaxed geodesic spanner}. This way, we obtain a
relaxed geodesic $6k$-spanner of size~$O(n\log^2n)$ and
complexity $O(m n^{1/k} + n\log^2 n)$ that can be computed in
$O((n+m)\log^2n \log m + K)$ time. Because our edges always consist
of at most three shortest paths, we can again output the edges
implicitly in $O(n \log^2 n)$ time. 

In Section~\ref{sec:simple_polygon_eps}, we use a relaxed geodesic spanner in a simple polygon to improve the spanning ratio of the low complexity spanner from $2\sqrt{2}k$ to $2k+\eps$, for any constant $\eps \in (0,2k)$. To achieve this, we adapt the refinement suggested by Abam, de Berg, and Seraji~\cite{SpannerPolyhedralTerrain} in such a way that the spanning ratio is reduced while the complexity increases only by a constant factor (dependent on $\varepsilon$). 
In the preliminary version of this paper~\cite{complexity_spanners} we claimed that application of this refinement was straightforward and that it could also be applied in a polygonal domain. However, this no longer guarantees bounded complexity. Our new approach does provide
the desired complexity bound, but only in case the domain is a
simple polygon. We leave obtaining a $(2k+\eps)$-spanner of the same complexity in a
polygonal domain as an open problem.

Finally, in Section~\ref{sec:lower_bounds}, we provide lower bounds on the complexity of (relaxed) geodesic spanners. For any constant $\varepsilon \in (0,1)$ and integer constant $t \geq 2$, we show a lower bound for the complexity of a $(t-\varepsilon)$-spanner in a simple polygon of $\Omega(mn^{1/(t-1)} + n)$.  
Therefore, the $2k+ \eps$ spanning ratio of our $O(mn^{1/k} + n\log^2n)$ complexity spanner is about a factor two off optimal.
For the case of a $(3-\eps)$-spanner, we prove an even stronger lower bound of $\Omega(nm)$.

Throughout the paper, we make the general position assumption that all
vertices of $P$ and sites in $S$ have distinct $x$- and
$y$-coordinates. Symbolic perturbation, in particular a shear transformation, can be used to remove this assumption~\cite{Comp_geom_book}.

%%%%%%%%%%%%%%%%%%%%%%%%%%%%%%%%%%%%%%%%%%%%%%%%%%%%%%%%%%%%%%%%%%%
\section{A 1-dimensional additively weighted 2-spanner}\label{sec:1D-spanner}
We consider how to compute an additively weighted spanner $\G$ in 1-dimensional Euclidean space, where each site $p \in S$ has a non-negative weight $w(p)$. The distance $d_w(p,q)$ between two sites $p,q \in S$ is given by $d_w(p,q) = w(p) + |pq| + w(q)$, where $|pq|$ denotes the Euclidean distance. Without loss of generality, we can map $\mathbb{R}^1$ to the $x$-axis, and the weights to the $y$-axis, see Figure~\ref{fig:1D_spanner}. This allows us to speak of the sites left (or right) of some site $p$. 

To construct a spanner $\G$, we first partition the sites into two sets $S_\ell$ and $S_r$ of roughly equal size by a point $O$ with $w(O) = 0$. The set $S_\ell$ contains all sites left of $O$, and $S_r$ all sites right of $O$. Sites that lie on the vertical line through $O$ are not included in either of the sets. We then find a site $c \in S$ for which $d_w(c, O)$ is minimal. For all $p\in S$, $p \neq c$, we add the edge $(p,c)$ to $\G$. Finally, we handle the sets $S_\ell$ and $S_r$, excluding the site $c$, recursively. 

\begin{figure}
    \centering
    \includegraphics{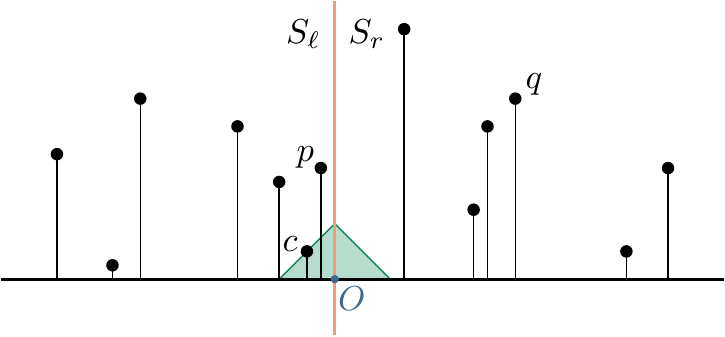}
    \caption{Construction of the additively weighted 1-dimensional spanner. The green triangle represents all points that are at distance at most $d_w(c,O)$ from $O$.}
    \label{fig:1D_spanner}
\end{figure}

\begin{lemma}\label{lem:1-dim-2-spanner}
The graph $\G$ is a 2-spanner of size $O(n\log n)$ and can be constructed in $O(n \log n)$ time.
\end{lemma}
\begin{proof}
As we add $O(n)$ edges in each level of the recursion, the total number of edges in $\G$ is $O(n\log n)$. Consider two sites $p,q \in S$. Let $c$ be the chosen center at the level of the recursion where $p$ and $q$ are assigned to different subsets $S_\ell$ and $S_r$. Assume without loss of generality that $p \in S_\ell$ and $q \in S_r$. Note that, because $p \in S_\ell$ and $q \in S_r$ we have $d_w(p,q) = d_w(p,O) + d_w(q,O)$. Furthermore, $d_w(c,O) \leq d_w(p,O)$ and $d_w(c,O) \leq d_w(q,O)$, by the choice of $c$. Because both edges $(p,c)$ and $(q,c)$ are in $\G$, we get for $d_\G(p,q)$:
\begin{equation}\label{eq:dw}
    d_\G(p,q) \leq d_w(p,O) + 2d_w(c,O) + d_w(q,O) \leq 2d_w(p,O) + 2d_w(q,O) = 2d_w(p,q).
\end{equation}
When there is no such center point, so $p$ or $q$ lies on the vertical line through $O$ at some level, then it still holds for this level that $d_w(p,q) = d_w(p,O) + d_w(q,O)$. Equation~(\ref{eq:dw}) again gives that $d_\G(p,q) \leq 2d_w(p,q)$.

We can find a point $O$ that separates the points into two sets $S_\ell$ and $S_r$ of equal size at each level of the recursion in linear time. Additionally, a linear number of edges is added to the spanner at each level. The running time is thus $O(n\log n)$. 
\end{proof}

%%%%%%%%%%%%%%%%%%%%%%%%%%%%%%%%%%%%%%%%%%%%%%%%%%%%%%%%%%%%%%%%%%%
\section{Spanners in a simple polygon}\label{sec:simple_polygon}
\subsection{A simple geodesic spanner}\label{sec:simple_geodesic_spanner}
Just like Abam, de Berg, and Seraji~\cite{SpannerPolyhedralTerrain},
we use our 1-dimensional spanner to
construct a geodesic spanner. We are more interested in the simplicity of the spanner than its spanning ratio, as we base our low complexity spanners, to be discussed in Section~\ref{sec:complexity_spanners}, on this simple geodesic spanner. 
 Let $P$ be a simple polygon, and let $\partial P$ denote the polygon boundary. We denote by $d(p,q)$ the geodesic
distance between $p,q \in P$, and by $\pi(p,q)$ the shortest
(geodesic) path from $p$ to~$q$. We analyze the simple construction using any 1-dimensional additively weighted $t$-spanner of size
$O(n\log n)$. We show that restricting the domain to a simple polygon improves the
spanning ratio from $3t$ to $\sqrt{2}t$. The construction can be refined to achieve a spanning ratio of $t + \varepsilon$, see Section~\ref{sec:refinement}.

As in \cite{SpannerPolygonalDomain} and \cite{SpannerPolyhedralTerrain}, we first partition $P$ into two subpolygons $P_\ell$ and $P_r$ by a line segment $\lambda$, such that each subpolygon contains at most two thirds of the sites in $S$~\cite{PolygonCutting}. We assume, without loss of generality, that $\lambda$ is a vertical line segment and $P_\ell$ is left of $\lambda$. Let $S_\ell$ be the sites in the closed region $P_\ell$, and $S_r := S \setminus S_\ell$. For each site $p \in S$, we then find the point $p_\lambda$ on $\lambda$ closest to $p$. We call $p_\lambda$ the projection of $p$. Note that this point is unique, because the shortest path to a line segment is unique in a simple polygon. We denote by $S_\lambda$ the set of all projected sites. As $\lambda$ is a line segment, we can define a weighted 1-dimensional Euclidean space on $\lambda$, where $w(p_\lambda):=d(p,p_\lambda)$ for each $p_\lambda \in S_\lambda$. We compute a $t$-spanner $\G_\lambda = (S_\lambda, E_\lambda)$ for this set. For each pair $(p_\lambda,q_\lambda) \in E_\lambda$, we add the edge $(p,q)$, which is $\pi(p,q)$, to our spanner $\G$. Finally, we recursively compute spanners for $S_\ell$ and $S_r$, and add their edges to $\G$ as well.

\begin{lemma}\label{lem:1D_to_geodesic_spanner}
The graph $\G$ is a geodesic $\sqrt{2}t$-spanner of size $O(n\log^2 n)$.
\end{lemma}

\begin{proof}
As $\G_\lambda$ has $O(n \log n)$ edges (Lemma~\ref{lem:1-dim-2-spanner}) that directly correspond to edges in $\G$, and the recursion has $O(\log n)$ levels, we have $O(n \log^2n)$ edges in total.
Let $p,q$ be two sites in $S$. If both are in $S_\ell$ (or $S_r$), then there is a path of length $\sqrt{2}t d(p,q)$ by induction. So, we assume w.l.o.g. that $p \in S_\ell$ and $q \in S_r$. Let $r$ be the intersection point of $\pi(p,q)$ and $\lambda$. Observe that $p_\lambda$ and
$q_\lambda$ must be on opposite sides of $r$, otherwise $r$ cannot be on the
shortest path. We assume, without loss of generality, that $p_\lambda$ is above $r$ and $q_\lambda$ below $r$. 
Because $\G_\lambda$ is a $t$-spanner, we know that there is a weighted path from $p_\lambda$ to $q_\lambda$ of length at most $td_w(p_\lambda,q_\lambda)$. As $w(p_\lambda) = d(p,p_\lambda)$, this directly corresponds to a path in the polygon. So, 
\begin{equation}\label{eq:spanner_distance}
    d_\G(p,q) \leq d_{\G_\lambda}(p_\lambda,q_\lambda) \leq td_w(p_\lambda,q_\lambda) = t(d(p,p_\lambda) + |p_\lambda r| + |rq_\lambda| + d(q_\lambda,q)).
\end{equation}

\begin{figure}
    \centering
    \includegraphics{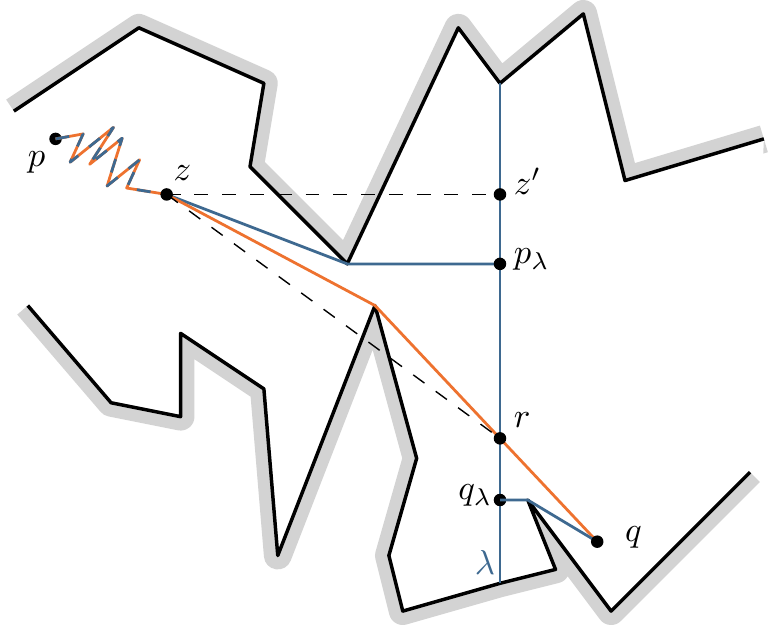}
    \caption{The shortest path $\pi(p,q)$ crosses $\lambda$ at $r$. The difference in length between the direct path from $z$ to $r$ and the path through $p_\lambda$ can be bounded by considering the triangle $\T=(z,z',r)$.}
    \label{fig:2sqrt2_spanner}
\end{figure}
Let $z$ be the point where the shortest paths from $p$ to $p_\lambda$
and $r$ separate. See Figure~\ref{fig:2sqrt2_spanner} for an
illustration. Consider the right triangle $\T = (z, z', r)$, where $z'$ is the intersection point of the line perpendicular to $\lambda$ through $z$ and the line containing $\lambda$. Note that $z'$ does not necessarily lie within $P$. For this triangle we have that 
\begin{equation}\label{eq:triangle_bounds}
    |zr| \geq \frac{\sqrt{2}}{2} (|zz'| + |z'r|).
\end{equation}

Next, we show that the path from $z$ to $p_\lambda$ is a $y$-monotone convex polygonal chain ending at or below $z'$. Consider the vertical ray through $z$ upwards to the polygon boundary. We call the part of $\partial P$ between where the ray hits $\partial P$ and $\lambda$ the \emph{top} part of~$\partial P$. Similarly, for a downwards ray, we define the \emph{bottom} part of $\partial P$. There are no vertices on $\pi(z,p_\lambda)$ from the bottom part of $\partial P$, because such a vertex would then also occur on the shortest path to $r$. This is in contradiction with the definition of $z$. If $z$ sees $z'$, then $p_\lambda = z'$, otherwise the chain must bend at one or more vertices of the top part of $\partial P$, and thus lie below $z'$. It follows that $\pi(z,p_\lambda)$ is contained within \T. 
Similarly, we conclude that $\pi(z,r)$ is contained within~$\T$. Additionally, this gives us that $d(z,p_\lambda) \leq |zz'| + |z'p_\lambda|$, and $d(z,r) \geq |zr|$. Together with Equation~(\ref{eq:triangle_bounds}) this yields $d(z,p_\lambda) + |p_\lambda r| \leq |zz'| + |z'r| \leq \sqrt{2}|zr| \leq \sqrt{2} d(z,r)$. And thus 
\begin{equation*}
    d(p,p_\lambda) + |p_\lambda r| = d(p,z) + d(z,p_\lambda) + |p_\lambda r| \leq d(p,z) + \sqrt{2} d(z,r) \leq \sqrt{2} d(p,r).
\end{equation*}
Symmetrically, we find for $q$ that $d(q,q_\lambda) + |q_\lambda r| \leq \sqrt{2} d(q,r)$. From this, together with Equation~(\ref{eq:spanner_distance}), we conclude that
$d_\G(p,q) \leq t\left(\sqrt{2} d(p,r) + \sqrt{2} d(r,q)\right) = \sqrt{2}td(p,q)$.
\end{proof}

Applying Lemma~\ref{lem:1D_to_geodesic_spanner} to the spanner of Section~\ref{sec:1D-spanner} yields a  $2\sqrt{2}$-spanner of size $O(n\log^2n)$.

\subsubsection{A refinement to obtain a $(t +\eps)$-spanner}\label{sec:refinement}
Abam~\etal~\cite{SpannerPolyhedralTerrain} refine their spanner construction to obtain a $(2+\eps)$-spanner for any constant $\eps > 0$. In the following lemma, we apply their refinement to the construction of the spanner proposed in Section~\ref{sec:simple_geodesic_spanner} and obtain a $(t + \eps)$-spanner.

\begin{lemma}\label{lem:refinement}
    Using any 1-dimensional additively weighted $t$-spanner of size $O(n\log n)$, we can construct a $(t+\varepsilon)$-spanner for a set $S$ of $n$ point sites that lie in a simple polygon $P$ of size $O(c_{\varepsilon,t} n \log^2 n)$, where $c_{\varepsilon,t}$
is a constant depending on~$\varepsilon$ and~$t$. The number of sites used to construct the 1-dimensional spanner is $O(c_{\varepsilon,t} n)$.
\end{lemma}

\begin{proof}
In the refined construction, instead of adding only a single point $p_\lambda$ to $S_\lambda$ for each site $p$, we additionally add a collection of $O(1/\delta^2)$ points on $\lambda$ ``close'' to $p_\lambda$ to $S_\lambda$, where $\delta$ is a constant depending on $\varepsilon$. These additional points all lie within distance $(1+2/\delta)\cdot d  (p, p_\lambda )$ of $p_\lambda$. 
The points are roughly equally spaced on the line segment within this distance. To be precise, the segment is partitioned into $O(1/\delta^2)$ pieces of length $\delta \cdot d(p,p_\lambda)$, and for each piece $i$ the point $p_\lambda^{(i)}$ closest to $p$ is added to $S_\lambda$. The weight of each point is again chosen as the geodesic distance to $p$. The 1-dimensional spanner(s) $\G_\lambda$ is then computed on this larger set $S_\lambda$. For each edge $(p_\lambda^{(i)}, q_\lambda^{(j)})$ in $\G_\lambda$, we again add the edge $(p,q)$ to the final spanner $\G$.

Abam, de Berg, and Seraji prove that for each $p,q \in S$, there are
points $p_\lambda^{(i)}, q_\lambda^{(j)} \in S_\lambda$ such that
$d_w(p_\lambda^{(i)}, q_\lambda^{(j)}) \leq  (1+\delta) \cdot
d(p,q)$. As $\G_\lambda$ is a $t$-spanner for $S_\lambda$, choosing
$\delta = \varepsilon/t$ implies that $d_\G(p,q) \leq t \cdot d_w(p_\lambda^{(i)}, t_\lambda^{(j)}) \leq t(1+\delta)d(p,q) = (t+\varepsilon)d(p,q)$. In other words, $\G$ is a $(t + \varepsilon)$-spanner of $S$. Note that the number of edges in the
spanner has increased because we use $O(n/\delta^2)$ instead of
$O(n)$ points to compute the 1-dimensional spanner. This results in a
spanner of size $O(c_{\varepsilon,t} n \log^2 n)$, where $c_{\varepsilon,t} = O(t^2/\varepsilon^2)$
is a constant depending on~$\varepsilon$.
\end{proof}

%%%%%%%%%%%%%%%%%%%%%%%%%%%%%%%%%%%%%%%%%%%%%%%%%%%%%%%%%%%%%%%%%%%
\subsection{Low complexity geodesic spanners}\label{sec:complexity_spanners}

In general, a geodesic spanner $\G = (S,E)$ in a simple polygon with $m$ vertices may have complexity $O(m|E|)$. It is easy to see that the $2\sqrt{2}$-spanner of Section~\ref{sec:simple_geodesic_spanner} can have complexity $\Omega(nm)$, just like the spanners in~\cite{SpannerPolyhedralTerrain}. As one of the sites, $c$, is connected to all other sites, the polygon in Figure~\ref{fig:lowerbound_3spanner} provides this lower bound. The construction in Figure~\ref{fig:lowerbound_3spanner} even shows that the same lower bound holds for the complexity of any $(3-\epsilon)$-spanner. Additionally, the following theorem implies a trade-off between the spanning ratio and the spanner complexity.

\begin{restatable}{theorem}{lowerBound}\label{thm:general_lower_bound}
For any constant $\varepsilon \in (0,1)$ and integer constant $t \geq 2$, there exists a set of $n$ point sites in a simple polygon $P$ with $m = \Omega(n)$ vertices for which any (relaxed) geodesic $(t-\varepsilon)$-spanner has complexity $\Omega(mn^{1/(t-1)})$. 
\end{restatable}

The proofs of these lower bounds are in Section~\ref{sec:lower_bounds}. Next, we present a spanner that almost matches this bound. We first present a $4\sqrt{2}$-spanner of bounded complexity, and then generalize the approach to obtain a $2\sqrt{2}k$-spanner of complexity $O(mn^{1/k} + n\log^2 n)$, for any integer $k \geq 2$. In Section~\ref{sec:simple_polygon_eps}, we show how to improve the spanning ratio of this spanner to $2k + \eps$, for any $\eps \in (0,2k)$.

\begin{figure}
    \centering
    \includegraphics{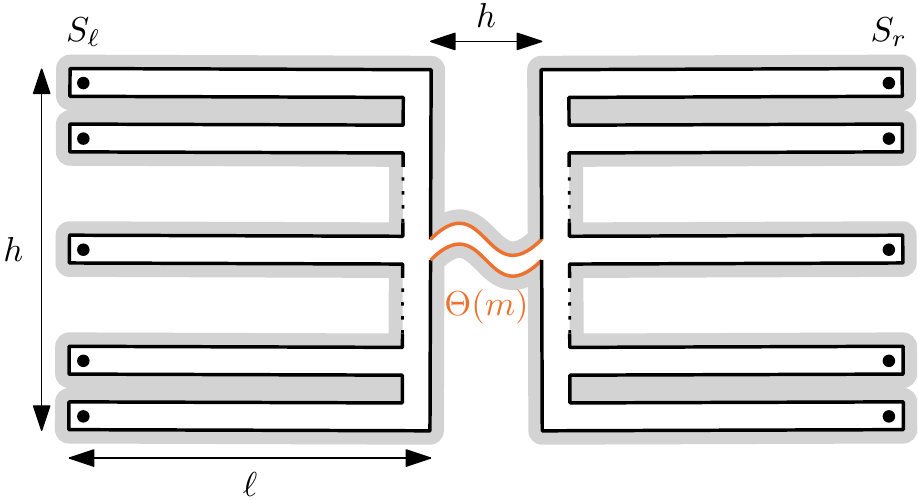}
    \caption{A simple polygon with a set of sites for which any $(3-\varepsilon)$-spanner has complexity~$\Omega(nm)$.}
    \label{fig:lowerbound_3spanner}
\end{figure}

\subsubsection{A $4\sqrt{2}$-spanner of complexity $O(m\sqrt{n} + n\log^2n)$}\label{sec:4sqrt2spanner}

To improve the complexity of the geodesic spanner, we adapt our construction for the additively weighted spanner $\G_\lambda$ as follows. After finding the site $c_\lambda\in S_\lambda$ for which $d_w(c_\lambda,O)$ is minimal, we do not add all edges $(p_\lambda,c_\lambda)$, $p_\lambda\in S_\lambda$, to $\G_\lambda$. Instead, we form groups of sites whose original sites (before projection) are ``close'' to each other in $P$. For each group $S_i$, we add all edges $(p_\lambda,c_{i,\lambda})$, $p_\lambda \in S_i$, to $\G_\lambda$, where $c_{i,\lambda}$ is the site in $S_i$ for which $d_w(c_{i,\lambda}, O)$ is minimal. Finally, we add all edges $(c_{i,\lambda}, c_\lambda)$ to~$\G_\lambda$.

To make sure the complexity of our spanner does not become too large, we must choose the groups in such a way that the edges in our spanner do not cross ``bad'' parts of the polygon too often. We achieve this by making groups of roughly equal size, where shortest paths within each group $S_i$ are contained within a region $R_i$ that is (almost) disjoint from the regions $R_j$ of other groups. We first show how to solve a recursion that is later used to bound the complexity of the spanner. The subsequent lemma formally states the properties that we require of our groups, and bounds the complexity of such a spanner.

\begin{lemma}\label{lem:1_dim_recursion}
    The recursion $T(n,m) = T(n/2, m_1) + T(n/2,m_2) + O(mn^{1/k} + n \log^{c_1}n)$, where $m_1 + m_2 = m + c_2$ and $c_1,c_2$ integer constants, solves to $T(n) = O(mn^{1/k} + n \log^{c_1 + 1} n)$.
\end{lemma}
\begin{proof}
We will write the recursion as a sum over all levels $i$ and all subproblems $j$ at each level. Because $n$ is halved in each level, the recursion has $O(\log n)$ levels. There are $2^i$ subproblems at level $i$ of the recursion. 

Let $m_{i,j}$ denote the $m$ value used in subproblem $j$ of level $i$. We consider the sum of the $m$ values over all subproblems at level $i$ which we denote by $M_i$, so $M_i = \sum_{j = 0}^{2^i} m_{i,j}$. We prove by induction that $M_i = m + c_2\cdot (2^{i}-1)$ for any $i \geq 1$. For $i = 1$ this states that $M_1 = m + c_2\cdot (2^1 - 1) = m + c_2$, which is equivalent to $m_1 + m_2 = m + c_2$. Suppose that the hypothesis $M_i = m + c_2 \cdot (2^i - 1)$ holds for $i = k$. The $2^{k+1}$ subproblems at level $k+1$ consist of $2^{k+1}/2 = 2^k$ pairs of subproblems $(j_1,j_2)$, each generated by a subproblem $j$ at level $k$, for which $m_{k+1,j_1} + m_{k+1,j_2} = m_{k,j} + c_2$. We can thus find $M_{k+1}$ by summing over these pairs. So, $M_{k+1} = \sum_{(j_1,j_2)} (m_{k,j} + c_2) = M_k + c_2\cdot 2^k = m + c_2 \cdot (2^k - 1)+c_2\cdot 2^k = m + c_2\cdot (2^{k+1} -1)$.

We are now ready to formulate the recursion as a summation. For simplicity we use that $M_i \leq m + c_2 2^i$.

\begin{align*}
    T(n,m) &= \sum_{i=0}^{O(\log n)} \sum_{j=0}^{2^i} \left( m_{i,j}\left(\frac{n}{2^i}\right)^{1/k} + \frac{n}{2^i} \log^{c_1}\left(\frac{n}{2^i}\right) \right) \\
    &=  \sum_{i=0}^{O(\log n)} \sum_{j=0}^{2^i}  m_{i,j}\left(\frac{n}{2^i}\right)^{1/k} + \sum_{i=0}^{O(\log n)} \sum_{j=0}^{2^i} \frac{n}{2^i}\log^{c_1}\left(\frac{n}{2^i}\right) \\
    &= \sum_{i=0}^{O(\log n)} \left(\frac{n}{2^i}\right)^{1/k} \sum_{j=0}^{2^i}  m_{i,j} + O(n\log^{c_1+1} n)\\
    &\leq \sum_{i=0}^{O(\log n)} \left(\frac{n}{2^i}\right)^{1/k} (m + c_2 2^i) + O(n\log^{c_1+1} n) \\
    &= mn^{1/k} \sum_{i=0}^{O(\log n)} \frac{1}{2^{i/k}} + c_2n^{1/k} \sum_{i=0}^{O(\log n)} 2^{(1-1/k)i} + O(n\log^{c_1+1} n)\\
    &= O(mn^{1/k}) + c_2n^{1/k} \cdot O(n^{1-1/k}) + O(n\log^{c_1+1} n)\\
    &= O(mn^{1/k} + n\log^{c_1+1} n). \qedhere
\end{align*}
\end{proof}

\begin{lemma}\label{lem:group_properties}
    If the groups adhere to the following properties, then $\G$ has $O(m\sqrt{n}+n\log^2n)$ complexity:
    \begin{enumerate}
        \item each group contains $\Theta(\sqrt{n})$ sites, and
        \item each vertex of $P$ is only used by shortest paths within $O(1)$ groups.
    \end{enumerate}
\end{lemma}

\begin{proof}
    We will first prove the complexity of the edges in one level of the 1-dimensional spanner is $O(m\sqrt{n} + n)$. Two types of edges are added to the spanner: (a) edges from some $c_i$ to $c$, and (b) edges from some $p \in S_i$ to $c_i$. According to property 1, there are $\Theta(\sqrt{n})$ groups, and thus $\Theta(\sqrt{n})$ type (a) edges, that each have a complexity of $O(m)$. Thus the total complexity of these edges is $O(m\sqrt{n})$. Let $r_i$ be the maximum complexity of a shortest path between any two sites in $S_i$ and let $V_i$ be the set of vertices this path visits. 
    Property 2 states that for any $v \in V_i$ it holds that $|\{j \mid v \in V_j\}| = O(1)$, which implies that $\sum_i r_i = O(m)$. The complexity of all type (b) edges is thus $O(n) + \sum_i r_i O(\sqrt{n}) = O(m\sqrt{n} + n)$. 
    
    Next, we show that in both recursions, the 1-dimensional recursion and the recursion on $P_\ell$ and $P_r$, not only the number of sites, but also the complexity of the polygon is divided over the two subproblems. Splitting the sites into left and right of $O$ corresponds to splitting the polygon horizontally at $O$: all sites left (right) of $O$ in the 1-dimensional space lie in the part of the polygon below (above) this horizontal line segment. Thus, shortest paths between sites left of $O$ use part of the polygon that is disjoint from the shortest paths between the sites right of $O$. This means that for two subproblems we have that $m_1 + m_2 = m$, where $m_i$ denotes the maximum complexity of a path in subproblem $i$. The recursion for the complexity is now given by
    \begin{equation*}
    T(n,m) = T(n/2, m_1) + T(n/2,m_2) + O(m\sqrt{n} + n)\text{, with } m_1 + m_2 = m.
    \end{equation*}
    According to Lemma~\ref{lem:1_dim_recursion} this solves to $T(n) = O(m\sqrt{n} + n \log n)$.
    
    Similarly, the split by $\lambda$ divides the polygon into two subpolygons, while adding at most two new vertices. As all vertices, except for the endpoints of $\lambda$, are in $P_\ell$ or $P_r$ (not both), the total complexity of both subpolygons is at most $m+4$. We obtain the following recursion
    \begin{equation*}
    T(n,m) = T(n/2, m_1) + T(n/2,m_2) + O(m\sqrt{n} + n\log n)\text{, with } m_1 + m_2 = m + 4.
    \end{equation*}
    Lemma~\ref{lem:1_dim_recursion} again states this solves to $T(n) = O(m\sqrt{n} + n\log^2n)$.
\end{proof}

To form groups that adhere to these two properties, we consider the shortest path tree $\mathit{SPT}_c$ of $c$: the union of all shortest paths from $c$ to the vertices of $P$.   We include the sites $p \in S\setminus \{c\}$ as leaves in $\mathit{SPT}_c$ as children of their apex, i.e., the last vertex on $\pi(c,p)$. This gives rise to an ordering of the sites in $S$, and thus of the weighted sites in $S_\lambda$, based on the in-order traversal of the tree. We assign the first $\lceil \sqrt{n} \rceil$ sites to $S_1$, the second $\lceil \sqrt{n} \rceil$ to $S_2$, etc. See Figure~\ref{fig:shortest_path_tree}.

\begin{figure}
    \centering
    \includegraphics{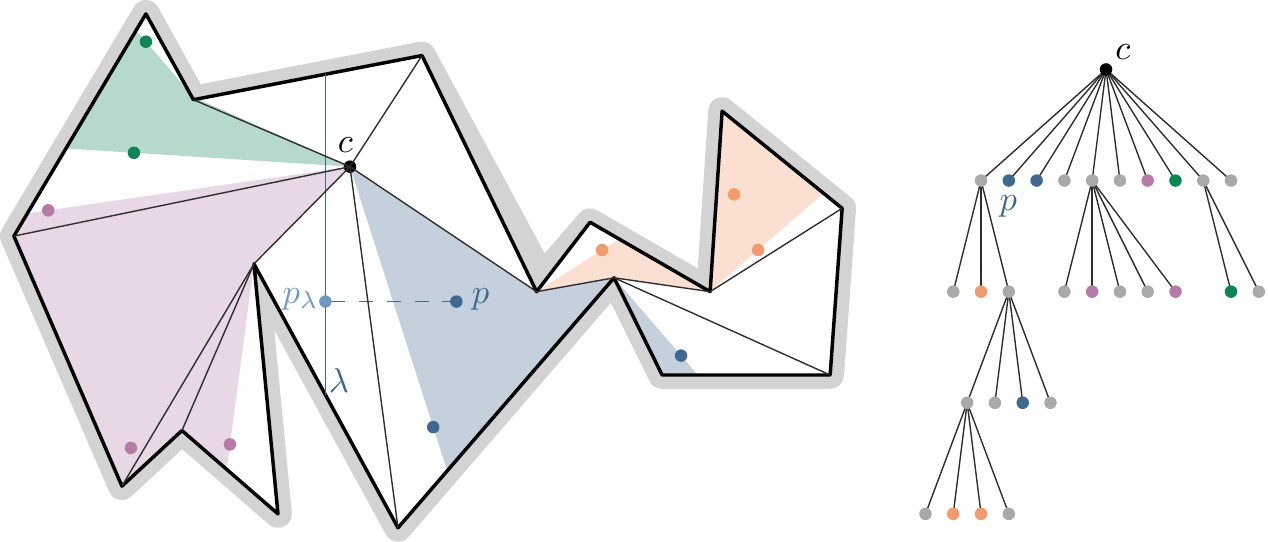}
    \caption{The shortest path tree of $c$. The polygon vertices are grey in the tree. Each group $S_i$ has an associated polygonal region $R_i$ in $P$.}
    \label{fig:shortest_path_tree}
\end{figure}

Clearly these groups adhere to property 1. Proving that they also adhere to property 2 is more involved. For each group $S_i$, consider the minimal subtree $\T_i$ of $\mathit{SPT}_c$ containing all $p \in S_i$. $\T_i$ defines a polygonal region $R_i$ in $P$ as follows. Refer to Figure~\ref{fig:shortest_path_tree} for an illustration. Let $v_i$ be the root of $\T_i$. Consider the shortest path $\pi(v_i,a)$, where $a$ is the first  site of $S_i$ in $\T_i$ by the ordering used before. Let $\pi_a$ be the path obtained from $\pi(v_i,a)$ by extending the last segment of $\pi(v_i,a)$ to the boundary of $P$. Similarly, let $\pi_b$ be such a path for the last site of $S_i$ in $\T_i$. 
We take $R_i$ to be the region in $P$ rooted at $v_i$ and bounded by $\pi_a$, $\pi_b$, and some part of the boundary of $P$, that contains the sites in $S_i$. In case $v_i$ is~$c$, we split $R_i$ into two regions $R_j$ and $R_k$, such that the angle of each of these regions at $c$ is at most $\pi$. The set $S_i$ is then also split into two sets $S_j$ and $S_k$ accordingly. The following three lemmas on $R_i$ and $\T_i$ together imply that the groups adhere to property 2.

\begin{restatable}{lemma}{OnlyTiInRi}
    Only vertices of $P$ that are in $\T_i$ can occur in $R_i$.
\end{restatable}
\begin{proof}
    Let $v$ be a vertex in $R_i$. The paths $\pi_a$ and $\pi_b$ are shortest paths from the polygon boundary to $v_i$. Additionally, we can extend these paths with $\pi(v_i,c)$ to obtain shortest paths to the site $c$. The shortest path from $v$ to $c$ cannot intersect either of these shortest paths twice, thus $v$ is in $\T_i$.
\end{proof}
   
\begin{restatable}{lemma}{PathsInSi}
    All shortest paths between sites in $S_i$ are contained within $R_i$.
\end{restatable} 
\begin{proof}
    Suppose there is a shortest path $\pi(p,q)$, $p_\lambda,q_\lambda \in S_i$, that is not contained within $R_i$. Then $\pi(p,q)$ must exit $R_i$ through $\pi_a$ and enter again through $\pi_b$, or the other way around. The path thus goes around $v_i$. Consider a line through $v_i$ that does not pass through the interior of $R_i$. Note that this exists because $v_i$ is not a reflex vertex of $R_i$. As $P$ is a simple polygon, the line must intersect $\pi(p,q)$ twice. This provides a shortcut of the shortest path by following along the line, which is a contradiction.
\end{proof}

\begin{restatable}{lemma}{occurenceSPT}\label{lem:occurence_spt}
    Any vertex $v \in \mathit{SPT}_c$ occurs in at most two trees $\T_i$ and $\T_j$ as a non-root node.
\end{restatable}
\begin{proof}
    Let $v \in \T$, and $\T(v)$ be the subtree rooted at $v$. Suppose $v$ is a non-root node of three trees $\T_i$, $\T_j$, and $\T_k$, $i < j < k$. Then there must be three sites $p_i \in S_i$, $p_j \in S_j$, and $p_k \in S_k$ in $\T(v)$. As our groups occur in order, we know $p_i$ is before $p_j$, and $p_j$ is before $p_k$ in the in-order traversal of $\mathit{SPT}_c$. Let $p$ be the parent of $v$. As $v$ is a non-root node of each of the subtrees, we know that $p$ is in $\T_i$, $\T_j$, and $\T_k$ as well. This implies that there must be a site $p_j' \in S_j$ in $\T \setminus \T(v)$. If $p_j'$ is before $p_j$, then $p_j'$ is also before $p_i$, because $p_j' \notin \T(v)$. This implies that $p_i$ is in $S_j$, because it lies between $p_j' \in S_j$ and $p_j \in S_j$, which is a contradiction. If $p_j'$ is after $p_j$, then the same reasoning implies $p_k \in S_j$, which is also a contradiction.
\end{proof}

Note that the root $r$ of $\T_i$ is never used in a shortest path between sites in $S_i$, because $r$ cannot be a reflex vertex of $R_i$. Consequently, Lemma~\ref{lem:group_properties} states that the spanner has complexity $O(m\sqrt{n}+n\log^2n)$.

\begin{lemma}\label{lem:4sqrt2-spanner}
The graph $\G$ is a geodesic $4\sqrt{2}$-spanner of size $O(n\log^2 n)$.
\end{lemma}
\begin{proof}
We prove the 1-dimensional spanner $\G_\lambda$ is a 4-spanner with $O(n\log n)$ edges. Together with Lemma~\ref{lem:1D_to_geodesic_spanner}, this directly implies $\G$ is a $4\sqrt{2}$-spanner with $O(n\log^2n)$ edges.

In each level of the recursion, we still add only a single edge for each site. Thus, the total number of edges is $O(n\log n)$. Again, consider two sites $p_\lambda,q_\lambda \in S_\lambda$, and let $c_\lambda$ be the chosen center point at the level where $p_\lambda$ and $q_\lambda$ are separated by $O$. Let $S_i$ be the group of $p_\lambda$ and $S_j$ the group of $q_\lambda$. Both the edges $(p_\lambda,c_{i,\lambda})$ and $(c_{i,\lambda},c_\lambda)$ are in $\G_\lambda$, similarly for~$q_\lambda$. 
We thus have a path $p_\lambda \to c_{i,\lambda} \to c_\lambda \to c_{j,\lambda} \to q_\lambda$ in $\G_\lambda$. Using that $d_w(p_\lambda,c_{i,\lambda}) \leq d_w(p_\lambda,O) + d_w(c_{i,\lambda},O)$, because of the triangle inequality, and $d_w(c_{i,\lambda},O) \leq d_w(p_\lambda,O)$, we find:

\begin{align*}
       d_{\G_\lambda}(p_\lambda,q_\lambda) &= d_w(p_\lambda,c_{i,\lambda}) + d_w(c_{i,\lambda},c_\lambda) + d_w(c_\lambda,c_{j,\lambda}) + d_w(c_{j,\lambda},q_\lambda)\\
    &\leq d_w(p_\lambda,O) + 2 d_w(c_{i,\lambda},O)+ 2d_w(c_\lambda,O) + 2d_w(c_{j,\lambda},O) + d_w(q_\lambda,O) \\
    &\leq 4d_w(p_\lambda,O) + 4d_w(q_\lambda,O)\\
    &= 4d_w(p_\lambda,q_\lambda)  \qedhere
\end{align*}\end{proof}

\subsubsection{A $2\sqrt{2}k$-spanner of complexity $O(mn^{1/k} +n\log^2n)$}\label{sub:general_complexity_spanner}

We first sketch how to generalize the approach of Section~\ref{sec:4sqrt2spanner} to obtain a spanner with a trade-off between the (constant) spanning ratio and complexity, and then formally prove the result in Lemma~\ref{lem:delta2sqrt2-spanner}. Fix $N = n^{1/k}$, for some integer constant $k \geq 1$. Instead of $\Theta(\sqrt{n})$ groups, we create $\Theta(N)$ groups. For each of these groups we select a center, and then partition the groups further recursively. By connecting each center to its parent center, we obtain a tree of height $k$. This results in a spanning ratio of $2\sqrt{2}k$.

\begin{restatable}{lemma}{simplepolygonspanner}\label{lem:delta2sqrt2-spanner}
For any integer constant $k \geq 1$, there exists a geodesic $2\sqrt{2}k$-spanner of size~$O(n\log^2n)$ and complexity $O(m n^{1/k} + n\log^2 n)$. 
\end{restatable}

\begin{proof}
We prove that the 1-dimensional spanner we just sketched is a $2k$-spanner of complexity $O(mn^{1/k} + n\log n)$ with $O(n\log n)$ edges. Lemma~\ref{lem:1D_to_geodesic_spanner}, together with Lemma~\ref{lem:1_dim_recursion}, then implies that $\G$ is a geodesic $2\sqrt{2}k$-spanner of complexity $O(mn^{1/k} + n\log^2n)$ and size $O(n\log^2n)$. 
We first describe the 1-dimensional spanner in more detail, then analyze its complexity, and finally analyze its size and spanning ratio.

Fix $N = n^{1/k}$. When building a single level of the 1-dimensional spanner $\G_\lambda$, instead of $\Theta(\sqrt{n})$ groups of size $\Theta(\sqrt{n})$, we create $\Theta(N)$ groups of size $\Theta(n^{1-1/k})$, based on the shortest path tree of $c$, for some integer constant $k \geq 1$. After selecting a center $c_i$ for each group $S_i$, we recursively split the groups further based on the shortest path trees of the new centers $c_i$, until we reach groups of size one. For each group $S_i$ at level $j$ of this recursion, a center $c^{(j)}_i$ is selected as the site in $S_i$ for which $d_w(c^{(j)}_i, O)$ is minimal. We add an edge from each $c^{(j)}_i$ to its parent $c^{(j-1)}_{i'}$. The final tree $T$ obtained this way has height $k$.

Let $r^{(j)}_i$ be the number of vertices in $R^{(j)}_i$ of a group $S_i$ at level $j$, where $R_i^{(j)}$ is defined as in Section~\ref{sec:4sqrt2spanner}. In other words, $r^{(j)}_i$ is the maximum complexity of a shortest path between two sites in $S_i$ at level $j$, as in the proof of Lemma~\ref{lem:group_properties}. 
Let $S_1^{(j)},...,S_{O(N)}^{(j)}$ be the subgroups of some group $S_{i'}^{(j-1)}$. Then the corresponding regions $R_1^{(j)},...,R_{O(N)}^{(j)}$ partition $R^{(j-1)}_{i'}$, and Lemma~\ref{lem:occurence_spt} implies that a vertex of $R_{O(N)}^{(j)}$ can be in at most two of the smaller regions. Thus, we still have that $\sum_i r^{(j)}_i = O(m)$. This implies that the complexity all edges from level $j$ to $j-1$ is $O(mn^{1/k} +n)$. As $T$ has height $O(k)$, the total complexity of the edges in $T$ is also $O(mn^{1/k}+ n)$. Lemma~\ref{lem:1_dim_recursion} implies that this results in a 1-dimensional spanner of complexity $O(mn^{1/k} + n\log n)$.

In a single level of the recursion to build the 1-dimensional spanner $\G_\lambda$, only one edge is added for each site, namely to its parent in the tree. We thus still add $O(n)$ edges in each level of the recursion, and $O(n\log n)$ edges in total.

Consider two sites $p,q \in S_\lambda$, and let $c$ be the chosen center point at the level where $p$ and~$q$ are separated by $O$. The spanning ratio of the 1-dimensional spanner is determined by the number of sites we visit on a path from $p \in S_\ell$ to $q \in S_r$. Observe that the height of the tree $T$ is $k$. We assume w.l.o.g. that the path from $p$ to $q$ is as long as possible, i.e. visits $2k - 1$ vertices. In a slight abuse of notation, we denote by $c^{(j)}_p$ the centers on the path in $T$ from $p$ to the root~$c$. Then there is a path $p \to c^{(k-1)}_p \to c^{(k-2)}_p \to ... \to c^{(0)} (= c) \to ... \to c^{(k-2)}_q + c^{(k-1)}_q \to q$. Using that $d_w(c^{(j)}_p,c^{(j+1)}_p) \leq d_w(c^{(j)}_p,O) + d_w(c^{(j+1)}_p,O)$ and that $d_w(c^{(j)}_p,O) \leq d_w(p,O)$, we find:

\begin{equation*}
\begin{split}
        d_{\G_\lambda}(p,q) &= d_w(p,c^{(k-1)}_p) + \sum_{j=0}^{k-2} d_w(c^{(j+1)}_p,c^{(j)}_p)  + \sum_{j = 0}^{k-2} d_w(c^{(j)}_q,c^{(j+1)}_q) + d_w(c^{(k-1)}_q,q)\\
    &\leq d_w(p,O) + 2 \sum_{j=1}^{k-1} d_w(c^{(j)}_p,O)+ 2d_w(c,O) +  2 \sum_{j = 1}^{k-1} d_w(c^{(j)}_q,O) + d_w(q,O) \\
    &\leq d_w(p,O) + 2(k-1)d_w(p,O) + 2d_w(c,O) + 2(k-1)d_w(q,O) + d_w(q,O)\\
    &\leq 2k(d_w(p,O) + d_w(q,O))\\
    &= 2kd_w(p,q).
\end{split}
\end{equation*}
Thus the spanning ratio of the 1-dimensional spanner is $2k$.
\end{proof}

%%%%%%%%%%%%%%%%%%%%%%%%%%%%%%%%%%%%%%%%%%%%%%%%%%%%%%%%%%%%%%%%%%%
\subsection{Construction algorithm}\label{sub:construction_algo_SP}
In this section we propose an algorithm to construct the spanners of Section~\ref{sec:complexity_spanners}. The following gives a general overview of the algorithm, which computes a $2\sqrt{2}k$-spanner in $O(n\log^2n + m\log n)$ time. 
In the rest of this section we will discuss the steps in more detail.
\begin{enumerate}
    \item Preprocess $P$ for efficient shortest path queries and build both the vertical decomposition \VD and horizontal decomposition \HD of $P$.
    \item For each $p \in S$, find the trapezoid in \VD and \HD that contains $p$. For each trapezoid $\nabla \in \VD$, store the number of sites of $S$ that lies in $\nabla$ and sort these sites on their $x$-coordinate.
    \item Recursively compute a spanner on the sites $S$ in $P$:
    \begin{enumerate}
        \item Find a vertical chord $\lambda$ of $P$ such that $\lambda$ partitions $P$ into two polygons $P_\ell$ and $P_r$, and each subpolygon contains at most $2n/3$ sites using the algorithm of Lemma \ref{lem:find_vertical_chord}.
        \item For each $p \in S$, find the point $p_\lambda$ on $\lambda$ and its weight using the algorithm of Lemma~\ref{lem:find_projections}, and add this point to $S_\lambda$.
        \item Compute an additively weighted 1-dimensional spanner $\G_\lambda$ on the set $S_\lambda$ using the algorithm of Lemma~\ref{lem:1-dim-spanner} or Lemma~\ref{lem:1-dim-spanner-bottom-up}.
        \item For every edge $(p_\lambda, q_\lambda) \in E_\lambda$ add the edge $(p,q)$ to $\G$.
        \item Recursively compute spanners for $S_\ell$ in $P_\ell$ and $S_r$ in $P_r$.
    \end{enumerate}
\end{enumerate}

In step 1, we preprocess the polygon in $O(m)$ time such that the distance between any two points $p,q \in P$ can be computed in $O(\log m)$ time~\cite{Chazelle_triangulate,2PSP_simple_polygon}. We also build the horizontal and vertical decompositions of $P$, and a corresponding point location data structure, as a preprocessing step in $O(m)$ time \cite{Chazelle_triangulate,Kirkpatrick_point_location}. We then perform a point location query for each site $p \in S$ in $O(n \log m)$ time in step 2 and sort the sites within each trapezoid in $O(n\log n)$ time in total. The following lemma describes the algorithm to compute a vertical chord that partitions $P$ into two subpolygons such that each of them contains roughly half of the sites in $S$. It is based on the algorithm of Bose~\etal~\cite{PolygonCutting} that finds such a chord without the constraint that it should be vertical. Because of this constraint, we use the vertical decomposition of $P$ instead of a triangulation in our algorithm.

\begin{restatable}{lemma}{VerticalChord}\label{lem:find_vertical_chord}
In $O(n + m)$ time, we can find a vertical chord of $P$ that partitions $P$ into two subpolygons $P_\ell$ and $P_r$, such that each subpolygon contains at most $2n/3$ sites of $S$.
\end{restatable}

\begin{proof}
Consider the dual tree of the vertical decomposition \VD. Because our
polygon vertices have distinct $x$- and $y$-coordinates, the maximum degree of any
node in the tree is four; at most two neighbors to the right and two to the left of the trapezoid.
We select an arbitrary node $r$ as root of the tree. For each node $v$, we compute $c(v)$: the number of sites in $S$ that lie in some trapezoid of the subtree rooted at $v$. These values can be computed in linear time (in the size of the polygon) using a bottom-up approach, because we already know the number of sites that lie in each trapezoid.

Let $v$ be an arbitrary node in the tree and $\nabla_v$ the corresponding trapezoid. We first show that there is a vertical segment contained in $\nabla_v$ that partitions the polygon such that each subpolygon contains at most $2/3$ of the sites if
\begin{enumerate}
    \item $n/3 \leq c(v) \leq 2n/3$, or
    \item $\nabla_v$ contains at least $2n/3$ sites, or
    \item for each child $w$ of $v$ we have $c(w) < n/3$, and $c(v) > 2n/3$.
\end{enumerate}

In case 1, we choose $\lambda$ as the vertical segment between the
trapezoid of $v$ and its parent. In case 2, we choose $\lambda$ as a
segment for which exactly $n/3$ sites in $\nabla_v$ lie left of
$\lambda$. As $\nabla_v$ contains more than $2n/3$ sites, at most
$n/3$ sites lie outside of $\nabla_v$, thus $P_\ell$ contains
between $n/3$ and $2n/3$ sites.

In case 3, we choose $\lambda$ to lie within $\nabla_v$. Note that $v$ must have at least two children, otherwise $\Delta_v$ would contain at least $2n/3-n/3 = n/3$ sites, thus we would be in case 1 or case 2. Assume that $\nabla_{p(v)}$ lies right of $\nabla_v$,
where $p(v)$ denotes the parent of $v$. We consider the two children $u,w$ for which the trapezoids lie left of $\nabla_v$, see Figure~\ref{fig:finding_lambda}. When there is only one such site $u$, we consider $c(w) = 0$. It holds that $c(u) + c(w) < 2n/3$. We choose $\lambda$ such that $\max(0, n/3 - c(u) - c(w))$ sites in $\nabla_v$ lie left of $\lambda$. Thus $P_\ell$ contains between $n/3$ and $2n/3$ sites of $S$. Similarly, we consider the children on the right when $\nabla_{p(v)}$ lies left of~$\nabla_v$.

\begin{figure}
    \centering
    \includegraphics{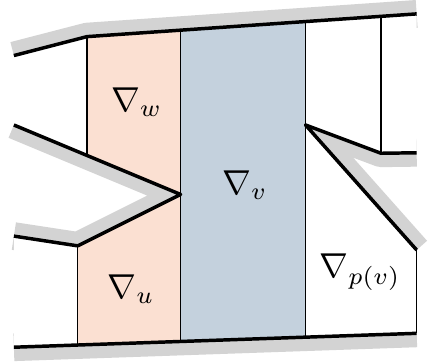}
    \caption{The trapezoid $\nabla_v$ has parent $\nabla_{p(v)}$ and relevant children $\nabla_u$ and $\nabla_w$.}
    \label{fig:finding_lambda}
\end{figure}

We now show that if none of the above conditions hold, there is a path in the tree to a node for which one of these conditions holds. If none of the above conditions hold, then either $c(v) < n/3$, or $c(v) > 2n/3$ and there is a child $w$ of $v$ with $c(w) \geq n/3$. In the first case, we consider the node $p(v)$, in the second case we consider the node $w$ with $c(w) \geq n/3$. By continuing like this, a path in the tree is formed. If the path ends up in $r$, it must hold that $c(r) > 2n/3$, and thus condition 2 holds. If the path ends up in a leaf node, then either condition 1 or 2 must hold for the leaf node. 

This proves not only that there exists such a vertical segment, but also provides a way to find such a segment. As the tree contains $O(m)$ nodes, we can find a trapezoid that can contain $\lambda$ in $O(m)$ time. We can separate the sites in a trapezoid by a vertical line segment such that exactly $x$ sites lie left of the segment in linear time. Thus, the algorithm runs in $O(n + m)$ time.
\end{proof}

The following lemma states that we can find the projections $p_\lambda$ efficiently. The algorithm produces not only these projected sites, but also the shortest path tree $\mathit{SPT}_\lambda$ of $\lambda$.

\begin{restatable}{lemma}{FindProjections}\label{lem:find_projections}
We can compute the closest point $p_\lambda$ on $\lambda$ and $d(p,p_\lambda)$ for all sites $p \in S$, and the shortest path tree $\mathit{SPT}_\lambda$, in $O(m + n\log m)$ time.
\end{restatable}

\begin{proof}
Consider the horizontal decomposition \HD of $P$ and its dual tree $T$. Let $x_\lambda$ denote the $x$-coordinate of the vertical line segment $\lambda$ and let $t$ be its top endpoint and $b$ its bottom endpoint. We choose the root $r$ of the tree as the trapezoid $\nabla_r$ for which $t \in \nabla_r$ and $\nabla_r \cap \lambda \neq \emptyset$. We color $T$ as follows, see Figure~\ref{fig:computing_plambda} for an example. Every node $v$ for which the corresponding trapezoid $\nabla_v$ contains the top or bottom endpoint of $\lambda$ is colored orange. Note that the root $r$ is thus colored orange. Every node $v$ for which $\nabla_v$ is crossed from top to bottom by $\lambda$ is colored blue. For every other node $v \in T$, consider its lowest ancestor $w$ that is colored either blue or orange. If $w$ is blue, then $v$ is colored green, if $w$ is orange, $v$ is colored purple. Next, we describe how to find $p_\lambda$ for a site $p  = (x_p, y_p) \in \nabla_v$ for each color of $v$.

\begin{figure}
    \centering
    \includegraphics{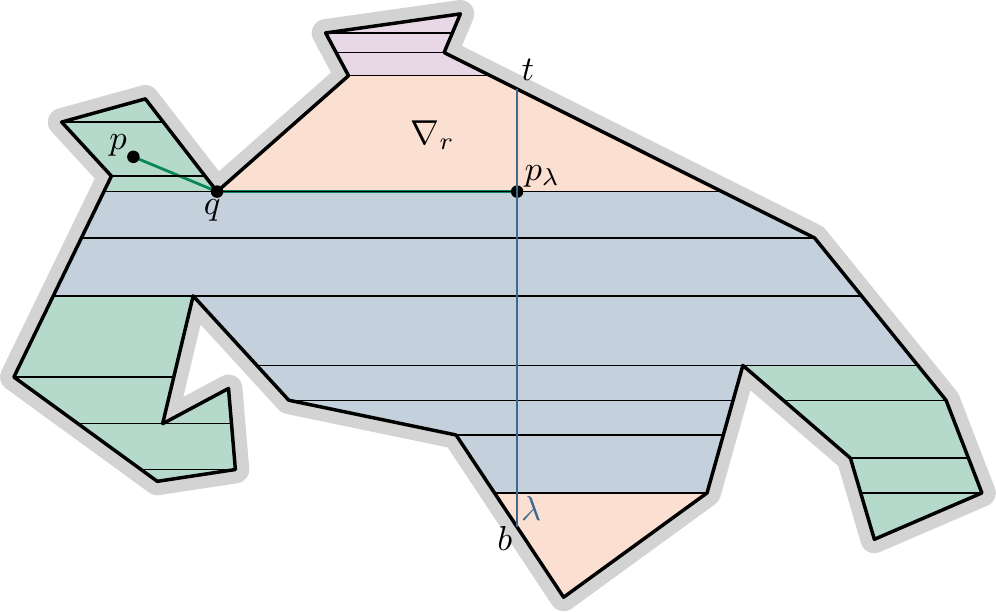}
    \caption{An example of the coloring of the horizontal decomposition in a simple polygon.}
    \label{fig:computing_plambda}
\end{figure}

\textit{Blue:} The horizontal line segment connecting $p$ and $\lambda$ is contained within $P$. So, $p_\lambda = (x_\lambda, y_p)$. 

\textit{Orange:} The shortest path to $\lambda$ is again a line segment. If $p$ lies above $t$ or below $b$, then $p_\lambda$ is $t$ or $b$, respectively. Otherwise,  $p_\lambda = (x_\lambda, y_p)$.

\textit{Green:} Let $w$ be the highest green ancestor of $v$ (possibly $v$ itself). Suppose that $\nabla_w$ lies above the trapezoid $\nabla_{p(w)}$ of its (blue) parent and $\nabla_w$ lies in $P_\ell$. Let $q$ be the bottom right corner of $\nabla_w$. Then $q_\lambda = (x_\lambda, y_q)$, as $q$ also lies in $\nabla_{p(w)}$. We will show that $p_\lambda = q_\lambda$. Suppose to the contrary that $p_\lambda \neq q_\lambda$. As $v$ is a descendant of $w$, any path from $p$ to $\lambda$ must intersect the bottom of $\nabla_v$. Let $q'$ be the this intersection point.
Because $q'$ lies on the bottom boundary of $\nabla_w$, the horizontal line segment $q'q_\lambda$ is contained within $P$ and thus $q'_\lambda = q_\lambda$. 
Because $d(q', q'_\lambda) = d(q', q_\lambda) < d(q', p_\lambda)$, subpath optimality implies that $d(p, q_\lambda) < d(p,p_\lambda)$, which is a contradiction. So $p_\lambda = q_\lambda$. Symmetrically, we can show a similar statement when $\nabla_v$ lies below $\nabla_{p(v)}$ and/or $\nabla_v$ lies in $P_r$.

\textit{Purple:} Let $w$ be the highest purple ancestor of $v$. Assume that $\nabla_{p(w)}$ lies below $\nabla_w$. Note that the bottom segment of $\nabla_w$ lies above $t$. It follows that for all points $q$ on this segment we have $q_\lambda = t$. According to the same argument as for the green trapezoids, we thus have $p_\lambda = t$. Symmetrically, if $\nabla_{p(w)}$ lies above $\nabla_w$, then $p_\lambda = b$.

We can thus find $p_\lambda$ for all $p \in S$ as follows. We perform a depth first search on $T$ starting at $r$. When we visit a node $v$, we first determine its color in $O(1)$ time. Then for all $p \in \nabla_v$ we determine $p_\lambda$ as described before. Note that this can also be done in constant time, because for a green/purple node we already computed the projected sites for the parent trapezoid. This thus takes $O(m + n)$ time overall. 

The shortest path tree $\mathit{SPT}_\lambda$ can now be computed as follows. All vertices of $P$ in a blue trapezoid are children of $\lambda$. At most four of these vertices lie on the boundary of the blue and green region (such as $q$ in Figure~\ref{fig:computing_plambda}). For each such vertex $q$, we include the shortest path tree of $q$ restricted to its respective green region as a subtree of $q$. We include the shortest path trees of $t$ and $b$ restricted to the top and bottom orange and purple regions as a subtree of $\lambda$. Because the regions where we construct a shortest path trees are disjoint, we can compute all of these shortest path trees in $O(m)$ time~\cite{Linear_time_spt}.
Finally, for all $p \in S$ we compute $d(p,p_\lambda)$ in $O(n\log m)$ time, and include each site in $\mathit{SPT}_\lambda$ as a child of their apex.
\end{proof}

\begin{lemma}\label{lem:1-dim-spanner}
    Given $\mathit{SPT}_\lambda$, we can construct a $4$-spanner $\G_\lambda$ on the additively weighted points $S_\lambda$, where the groups adhere to the properties of Lemma~\ref{lem:group_properties}, in $O(n\log n + m)$ time.
\end{lemma}

\begin{proof}
The $4$-spanner of Section~\ref{sec:4sqrt2spanner} requires an additional step at each level of the recursion, namely the formation of $\Theta(\sqrt{n})$ groups. We first discuss the running time to construct a $4$-spanner when forming the groups as in Section~\ref{sec:4sqrt2spanner}, and then improve the running time by introducing a more efficient way to form the groups.

In Section~\ref{sec:4sqrt2spanner}, the groups are formed based on the shortest path tree of the site $c$. Building the shortest path tree, and a corresponding point location data structure, takes $O(m)$ time~\cite{Linear_time_spt}. Then, we perform a point location query for each site to find its apex in the  shortest path tree, and add the sites to the tree. These queries take $O(n \log m)$ time in total. We form groups based on the traversal of the tree. Note that we do not distinguish between sites with the same parent in the tree, as the tree $\T_i$ (and thus the region $R_i$) obtained contains the same vertices of $P$ regardless of the order of these sites. After obtaining the groups, we again add only $O(n)$ edges to the spanner. The overall running time of the algorithm is thus $O((m + n \log m) \log n)$.

This running time can be improved by using another approach to form the groups. To form groups that adhere to the properties of Lemma~\ref{lem:group_properties}, and thus result in a spanner of the same complexity, we can use any partition of $P$ into regions $R_i$, as long $R_i$ as contains $\Theta(\sqrt{n})$ sites and $\sum_i r_i = O(m)$. Next, we describe how to form such groups efficiently using $\mathit{SPT}_\lambda$.

We first define an ordering on the sites. This is again based on the traversal of some shortest path tree. Instead of considering the shortest path tree of a point site, we consider the shortest path tree $\mathit{SPT}_\lambda$ of $\lambda$. Again, all sites in $S$ are included in this shortest path tree. Additionally, we split the node corresponding to $\lambda$ into a node for each distinct projection point on $\lambda$ (of the vertices and the sites) and add an edge between each pair of adjacent points, see Figure~\ref{fig:spt_lambda}. We root the tree at the node corresponding to the bottom endpoint of~$\lambda$. Whenever a node $t$ on $\lambda$ has multiple children, in other words, when multiple sites are projected to the same point $t$, our assumption that all $y$-coordinates are distinct ensures that all these sites lie either in $P_\ell$ or $P_r$.

The groups are formed based on the in-order traversal of this tree, which can be performed in $O(m+n)$ time. As before, the first $\lceil \sqrt{n} \rceil$ are in $S_1$, the second in $S_2$, etc. The groups thus adhere to the first property. Next, we show they also adhere to the second property.

\begin{figure}
    \centering
    \includegraphics{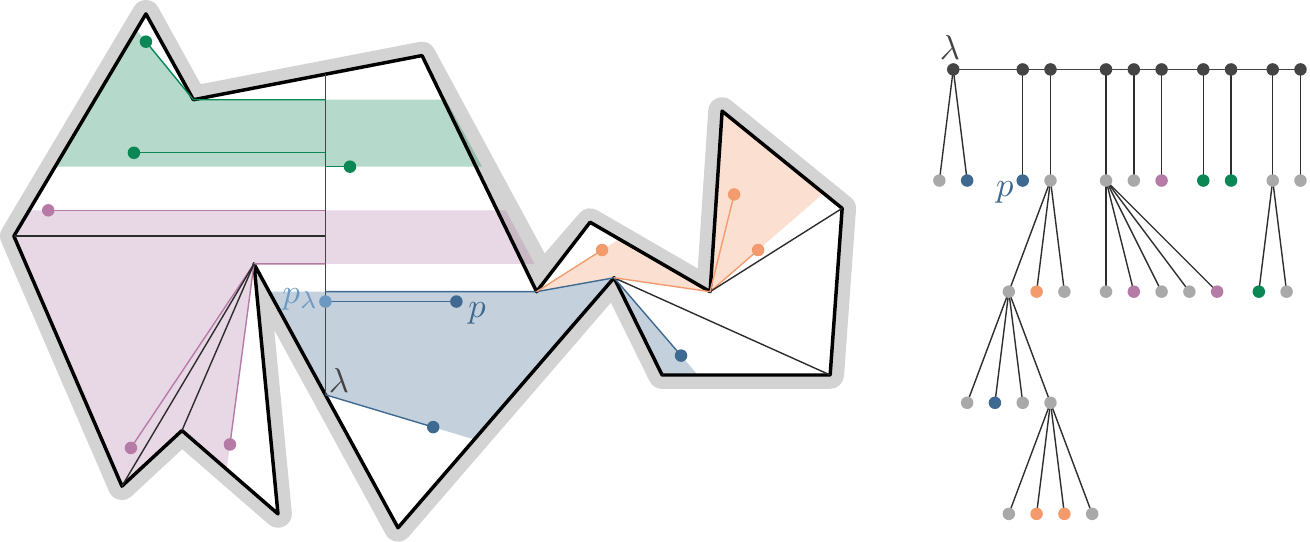}
    \caption{The shortest path tree $\mathit{SPT}_\lambda$ and associated polygonal region $R_i$ for each group $S_i$.}
    \label{fig:spt_lambda}
\end{figure}

For each group $S_i$, we again consider the minimal subtree $\T_i$ of $\mathit{SPT}_\lambda$ containing all $p \in S_i$. 
$\T_i$ defines a region $R_i$ in $P$ as follows. Let $a$ be the first site of $S_i$ in $\T_i$ by the ordering used before. Assume that $a$ lies in $P_\ell$. We distinguish two cases: $a_\lambda \in \T_i$, or $a_\lambda \notin \T_i$. When $a_\lambda \in \T_i$, then let $\pi_a$ be the path obtained from $\pi(a_\lambda,a)$ by extending the last segment to the boundary of $P$. Additionally, we extend $\pi_a$ into $P_r$ horizontally until we hit the polygon boundary. When $a_\lambda \notin \T_i$, consider the root $v_i$ of $\T_i$. Let $\pi_a$ be the path obtained from $\pi(v_i, a)$ by extending the last segment of the path to the boundary of~$P$. Similarly, let $\pi_b$ be such a path for the rightmost site of $S_i$ in $\T_i$. 
We take $R_i$ to be the region in $P$ bounded by $\pi_a$, $\pi_b$, and some part of the boundary of $P$ that contains the sites in $S_i$. See Figure~\ref{fig:spt_lambda}. Note that, as before, only vertices of $P$ that are in $\T_i$ can occur in $R_i$. All shortest paths between sites in $S_i$ are contained within $R_i$. Just as for the shortest path tree of $c$, Lemma~\ref{lem:occurence_spt} implies that any vertex $v \in \mathit{SPT}_\lambda$ occurs in at most two trees $\T_i$ and $\T_j$ as a non-root vertex. We conclude that any vertex is used by shortest paths within at most two groups.

After splitting $\lambda$ at a point $O$, the tree $\mathit{SPT}_\lambda$ is also split into two trees $\T_\ell$ and $\T_r$ that contain exactly the sites in $S_\ell$ and $S_r$. We can thus reuse the ordering to form groups at each level of the recursion.
This way, the total running time at a single level of the recursion is reduced to~$O(n)$. The overall running time thus becomes $O(n\log n + m)$.
\end{proof}

\begin{lemma}\label{lem:1-dim-spanner-bottom-up}
    Given $\mathit{SPT}_\lambda$, we can construct a $2k$-spanner $\G_\lambda$ on the additively weighted points $S_\lambda$, where groups are formed as in Section~\ref{sub:general_complexity_spanner}, in $O(n\log n +m)$ time.
\end{lemma}
\begin{proof}
    To construct the 1-dimensional $2k$-spanner of Section~\ref{sub:general_complexity_spanner}, we can use the shortest path tree of $\lambda$ to form the groups as before. Note that we can select a center for each group after computing the groups, as including the center in the subgroups does not influence spanning ratio or complexity. After ordering the sites based on the in-order traversal of $\mathit{SPT}_\lambda$, we can build  the tree of groups in linear time using a bottom up approach. As before, fix $N = n^{1/k}$. We first form the $\Theta(N^{k})$ lowest level groups, containing only a single site, and select a center for each group. Each group at level $i$ is created by merging $\Theta(N)$ groups at level $i-1$, based on the same ordering. We do not perform this merging explicitly, but for each group we select the site closest to $O$ of the merged level-($i-1$) centers as the center. Because our center property, being the closest to $O$, is decomposable, this indeed gives us the center of the entire group. This way, we can compute the edges added in one level of the recursion in linear time, so the running time remains $O(n\log n + m)$.
\end{proof}

The total running time thus becomes $O((n(\log n + \log m) + m)\log n) = O(n\log^2n + m\log n)$. Here, we used that $n \log n \log m = O(n\log^2n)$ for $m < n^2$, and $n \log n \log m = O(m \log n)$ for $m \geq n^2$. By splitting the polygon alternately based on the sites and the polygon vertices, we can replace the final $O(\log n)$ factor by $O(\min(\log n, \log m))$. 
Together with Lemma~\ref{lem:delta2sqrt2-spanner}, we obtain the following theorem. 

\begin{theorem}\label{thm:2sqrt2k-spanner-time}
Let $S$ be a set of $n$ point sites in a simple polygon $P$ with $m$ vertices, and let $k \geq 1$ be any integer constant. We can build a geodesic $2\sqrt{2}k$-spanner of size~$O(n\log^2n)$ and complexity $O(m n^{1/k} + n\log^2 n)$ in $O(n\log^2n + m\log n + K)$ time, where $K$ is the output complexity. 
\end{theorem}

%%%%%%%%%%%%%%%%%%%%%%%%%%%%%%%%%%%%%%%%%%%%%%%%%%%%%%%%%%%%%%%%%%%
\section{Balanced shortest-path separators}\label{sub:sp-separator}

Let $S$ be a set of $n$ point sites inside a polygonal domain $\P$ with $m$ vertices. We denote by $\partial \P$ the boundary the polygonal domain, i.e. the boundary of the outer polygon and the boundary of the holes.
In this section, we develop an algorithm to partition the polygonal domain $\P$ into two subdomains $\P_\ell$ and $\P_r$ such that roughly half of the sites in $S$ lie in $\P_\ell$ and half in $\P_r$. Additionally, we require that the curve bounding $\P_\ell$ consists of at most three shortest paths and possibly part of $\partial \P$.

In a polygonal domain, we cannot simply split the domain into two subdomains by a line segment, as we did to partition a simple polygon, because any line segment that appropriately partitions $S$ might intersect one or more holes. Even if we allow a shortest path between two points on the outer boundary of $\P$ as our separator, it is not always possible to split the sites into sets $S_\ell$ and $S_r$ of roughly equal size. See Figure~\ref{fig:sp_triangle_appendix} for an example. We thus need another approach for subdividing the domain. In the version of this paper published at SoCG 2023~\cite{complexity_spanners}, we adapted the balanced shortest-path separator of Abam, de Berg, and Seraji~\cite{SpannerPolyhedralTerrain} for this purpose. However, a much more efficient construction can be achieved by using a different approach based on the separator by Thorup~\cite{Thorup_separator_polygon}, which is the approach we take here.

\begin{figure}
    \centering
    \includegraphics{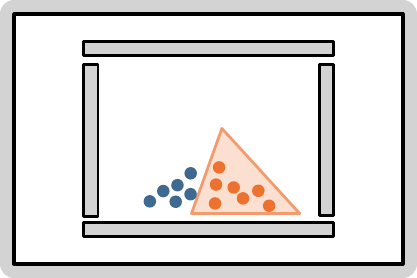}
    \caption{No shortest path between two points on the boundary of $\P$ can separate the sites into two groups. The sites can be separated by three shortest paths, for example using the orange triangle.}
    \label{fig:sp_triangle_appendix}
\end{figure}

To partition the polygonal domain into two subdomains $\P_\ell$ and $\P_r$, such that each contains roughly half of the sites in $S$, we allow three different types of separator. These separators consist of 1, 2, or 3 shortest paths, as seen in Figure~\ref{fig:separator_types_appendix}. Formally, we define a balanced shortest-path separator as follows.

\begin{definition}\label{def:sp-separator}
A \emph{balanced shortest-path separator} (sp-separator) partitions the polygonal domain into two subdomains $\P_\ell$ and $\P_r$, such that $2n/9 \leq |S_\ell| \leq 2n/3$. The separator is of one the following three types. 
\begin{itemize}
    \item A \emph{1-separator} consists of a single shortest path $\pi(u,v)$ that connects two points $u,v$ on the outer boundary of $P$. 
    \item A \emph{2-separator} consists of two shortest paths $\pi(u,v)$ and $\pi(u,w)$, where $u \in \P$ and $v,w \in \partial P$ on the boundary of the same hole or the outer polygon. 
    \item A \emph{3-separator} consists of three shortest paths $\pi(u,v)$, $\pi(v,w)$, and $\pi(w,u)$, where $u,v,w\in\P$, and each pair of shortest paths overlaps in a single interval, starting at the common endpoint. 
\end{itemize}
\end{definition}

\begin{figure}
    \centering
    \includegraphics{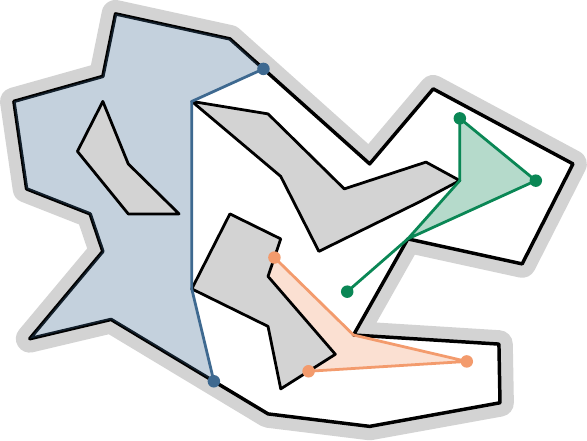}
    \caption{A 1-separator (blue), a 2-separator (orange), and a 3-separator (green).}
    \label{fig:separator_types_appendix}
\end{figure}

Note that one of the subpolygons $\P_\ell$ or $\P_r$ might be disconnected for a 2- or 3-separator.
We also allow a shortest path as a \emph{degenerate} 3-separator, $\P_\ell$ is then the degenerate polygon that is this shortest path. We call $u,v,w$ the \emph{corners} of the separator. Slightly abusing our notation, we refer to the closed subset of $\R^2$ corresponding to the polygonal domain $\P_\ell$ of a separator $\Delta$ by $\Delta$ itself. In the rest of this section, we prove the following theorem.

\begin{restatable}{theorem}{BalancedSeparatorImproved}\label{thm:find_separator_holes_improved}
Let $S$ be a set of $n$ point sites in a polygonal domain $\P$ with $m$ vertices. A balanced sp-separator exists and it can be computed in $O(m\log m + n\log n)$ time.
\end{restatable}

Thorup~\cite{Thorup_separator_polygon} proposed a shortest path separator for a polygonal domain $P$ consisting of at most six shortest paths (and possibly part of $\partial P$) that separates $P$ into three parts that each contain at most half of the triangles of some triangulation of $P$. His idea is to build the shortest path tree of an arbitrary vertex, and then apply his earlier algorithm to find a separator for a planar graph using a spanning tree of the graph~\cite{Thorup_separator_digraphs}. This approach is in turn based on the separator theorem for planar graphs by Lipton and Tarjan~\cite{separatorTheorem}. In our case, we want our separator not to separate the vertices of the polygon, but the sites in $S$. Next, we give a standalone explanation of how to handle this case, using the ideas of Thorup~\cite{Thorup_separator_polygon}, and Lipton and Tarjan~\cite{separatorTheorem}. Our goal is to find a balanced separator consisting of at most three shortest paths. In other words, we want to find a 1-, 2-, or 3-separator.

We start by constructing the shortest path tree $T$ of an arbitrary
vertex $v$ on the outer boundary of $P$. We then triangulate the space
both inside and outside the convex hull of $P$, including the interior
of the obstacles, constrained by the edges of $\partial P$ and
$T$. The triangulation within the convex hull of $P$ can be done using
straight line segments. However, we want to triangulate the space
outside of the convex hull as well, but for this we may use curved
edges, as illustrated in Figure~\ref{fig:separator_spt}. We call a triangle \emph{free} when it lies in $P$, so its interior is not part of a hole or the exterior. We assign each site to the triangle that contains the site. If there are multiple sites in the interior of the boundary of two triangles, then all of these sites are assigned to the same free triangle (chosen arbitrarily). 
% All sites on a vertex are assigned to the same (arbitrary) incident free triangle.
We set the weight of a triangle to the number of sites that is assigned to it. If there is a triangle of weight greater than $n/3$, then we can easily find a balanced separator that is contained within this triangle in $O(n\log n)$ time. When there are many sites on an edge of this triangle, our assignment ensures these are all assigned to the same triangle, which allows us to use a segment on this edge as a degenerate 3-separator. In the rest of this section, we thus assume that there is no such a heavy triangle.

We find a balanced separator by constructing a sequence $\Delta_0 \supset \Delta_1 \supset \dots \supset \Delta_k$ of separators such that $|\Delta_i \cap S| > 2n/3$ for $i =0,\dots k-1$ and $n/3 \leq |\Delta_k \cap S| \leq 2n/3$. Each separator $\Delta_i$ consists of an edge of the triangulation $(u_i,v_i) \notin T$ and the shortest paths $\pi(v,u_i)$ and~$\pi(v,v_i)$. See Figure~\ref{fig:separator_spt} for an illustration.

\begin{figure}
    \centering
    \includegraphics{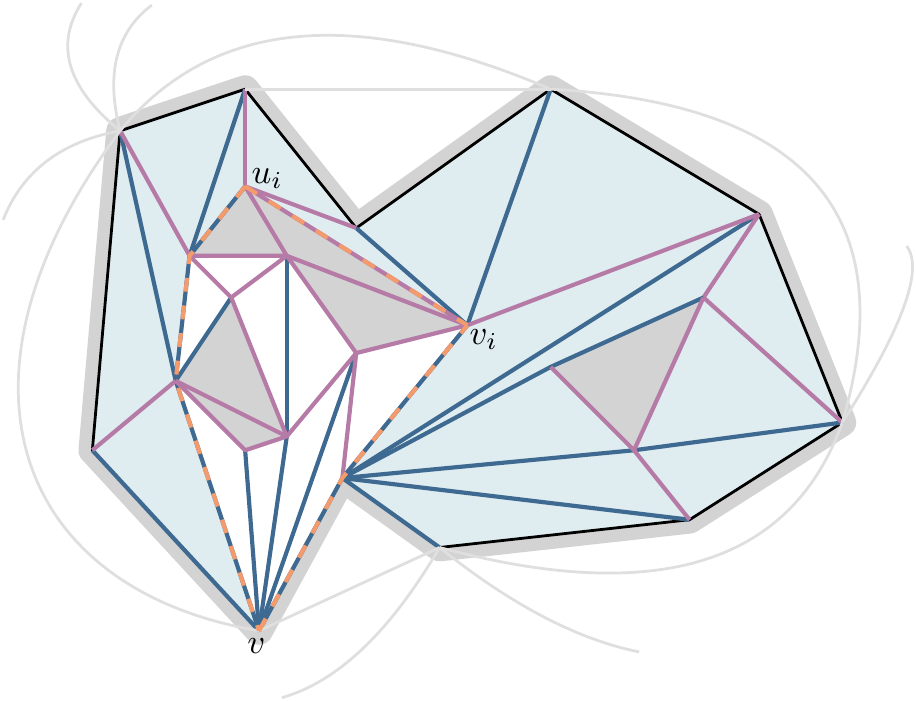}
    \caption{The shortest path tree $T$ of $v$ (blue edges) and a triangulation (purple edges inside $P$ and grey edges outside $P$). The separator $\Delta_i$ (dashed orange) separates the white part of the polygon from the blue part.}
    \label{fig:separator_spt}
\end{figure}

Let $(u_0,v_0)$ be an arbitrary non-tree edge. If $(u_0,v_0)$ is in the interior of a hole (or in the exterior), then the shortest paths $\pi(u_0,v)$ and $\pi(v_0,v)$, together with the obstacle boundary (or the boundary of the outer polygon), form a separator. To be precise, this is a 2-separator, unless $u_0$ or $v_0$ is equal to $v$ then it is a 1-separator. Otherwise, $\pi(u_0,v)$, $\pi(v_0,v)$, and $\pi(u_0,v_0)$ together form a 3-separator. Let $\Delta_0$ correspond to the side of the separator that has the highest weight. If the weight of $\Delta_0$ is at most $2n/3$, we are done. Otherwise, we find the next separator $\Delta_1$ as follows.

Consider the triangle $(u_0,v_0,w)$ adjacent to $(u_0,v_0)$ contained within $\Delta_0$. If one of the sides of this triangle is a tree edge, then let the non-tree edge be $(u_1,v_1)$. Otherwise, consider the separators defined by $(\pi(u_0,w), \pi(u_0, v), \pi(w,v))$ and $(\pi(v_0,w),\pi(v_0,v), \pi(w,v))$. Note that these separators are contained within $\Delta_0$, and their union is equal to $\Delta_0$. In this case, let $(u_1,v_1)$ be the edge $(u_0,w)$ or $(v_0,w)$ corresponding to the separator with largest weight. As the triangle $(u_0,v_0,w)$ has weight less than $n/3$, the weight of~$\Delta_1$ is at least $n/3$, and~$\Delta_1$ contains at least one less triangle than $\Delta_0$. We continue this process until we find a balanced separator $\Delta_k$. Note that the separator thus always consists of two shortest paths to $v$ from vertices of $P$ and is closed by either a line segment between two vertices of $P$, or part of the boundary. As the weight of the separator corresponds to the number of sites in $P_\ell$, this gives us a separator of the sites. 

Building the shortest path tree of $v$ and triangulating it takes $O(m\log m)$ time~\cite{Chazelle_triangulate,Linear_time_spt,Kirkpatrick_point_location}. We can compute the weight of all triangles in $O(n \log m)$ time using $O(n)$ point location queries. Next, we argue that the search for the separator can then be performed in $O(m)$ time~\cite{separatorTheorem}. Computing the weight of $\Delta_0$ can be done in $O(m)$, for example using a breadth-first search approach. The next separator $\Delta_{i+1}$ contains at least one triangle less than $\Delta_i$, so it can take at most $m$ steps to find a valid separator. In such a step there are two cases: either one side of the triangle $(u_i,v_i,w)$ is a tree edge, or none of its sides are. In the first case we simply compute the weight of $\Delta_{i+1}$ by subtracting the weight of the triangle $(u_i,v_i,w)$ from the weight of $\Delta_i$. In the second case, we must find the weight of the separators defined by $(\pi(u_i,w), \pi(u_i, v), \pi(w,v))$ and $(\pi(v_i,w),\pi(v_i,v), \pi(w,v))$. As they partition $\Delta_i$ it is sufficient to compute the weight of just one of these separators. We aggregate the weight of the two separators simultaneously, by using a breadth-first search in both problems and alternately taking one step in each problem. We stop whenever we find the complete weight of one of the separators. The number of step in the BFS is thus at most twice the size of the smallest of the two separators. As the triangles in one of these separators are not contained in any of the separators $\Delta_{i+1},\dots, \Delta_k$ (but are in $\Delta_i$), the running time over all separators $\Delta_0, \dots, \Delta_{k-1}$ is $O(m)$. We conclude that the total running time to compute a valid separator is $O(n \log n + m \log m + n \log m) = O(n\log n + m\log m)$.

%%%%%%%%%%%%%%%%%%%%%%%%%%%%%%%%%%%%%%%%%%%%%%%%%%%%%%%%%%%%%%%%%%%
\section{Spanners in a polygonal domain}\label{sec:polygonal_domain}
We consider a set of point sites $S$ that lie in a polygonal domain $\P$ with $m$ vertices and $h$ holes. Let $\partial \P$ denote the boundary of the outer polygon. In Section~\ref{sub:simple_geodesic_polygon}, we first discuss how to obtain a simple geodesic spanner for a polygonal domain, using the separator of Section~\ref{sub:sp-separator}. As before, the complexity of this spanner can be high. In Section~\ref{sub:12-spanner}, we discuss an adaptation to the spanner construction that achieves lower-complexity spanners, where the edges in the spanner are no longer shortest paths.

\subsection{A simple geodesic spanner}\label{sub:simple_geodesic_polygon}

A straightforward approach to construct a geodesic spanner for a polygonal domain would be to use the same construction we used for a simple polygon in Section~\ref{sec:simple_geodesic_spanner}. As discussed in Section~\ref{sub:sp-separator}, we cannot split the polygon into two subpolygons by a line segment $\lambda$. However, to apply our 1-dimensional spanner, we require only that the splitting curve $\lambda$ is a shortest path in $\P$. 
Instead of a line segment, we use the balanced sp-separator of Section~\ref{sub:sp-separator} to split the polygonal domain. There are three types of such a separator: 
a shortest path between two points on $\partial \P$ (1-separator), two shortest paths starting at the same point and ending at the boundary of a single hole (2-separator), three shortest paths $\pi(u,v)$, $\pi(v,w)$, and $\pi(u,w)$ with $u,v,w \in \P$ (3-separator). See Figure~\ref{fig:separator_types_appendix} for an illustration and Definition~\ref{def:sp-separator} for a formal definition. Let $\P_\ell$ be the polygonal domain to the left of $\lambda$, when $\lambda$ is a 1-separator, and interior to $\lambda$, when $\lambda$ is a 2- or 3-separator. Symmetrically, $\P_r$ is the domain to the right of $\lambda$ for a 1-separator and exterior to $\lambda$ for a 2-, or 3-separator. As before, let $S_\ell$ be the sites in the closed region $\P_\ell$, and $S_r := S\setminus S_\ell$. To compute a spanner on the set $S$, we project the sites to each of the shortest paths defining the separator, and consecutively run the 1-dimensional spanner algorithm once on each shortest path. Note that these projections are no longer unique, as there might be two topologically distinct shortest paths to $\lambda$. However, we can simply select one such that no two paths properly intersect to obtain the desired spanning ratio and spanner complexity. We then add the edge $(p,q)$ to our spanner $\G$ for each edge $(p_\lambda,q_\lambda)$ in the 1-dimensional spanners. Finally, we recursively compute spanners for the sites $S_\ell$ in $\P_\ell$ and $S_r$ in $\P_r$, just like in the simple polygon case.

Whenever the sp-separator intersects a single hole at two or more different intervals, then part of $\P_\ell$ or $\P_r$ becomes disconnected. When this happens, we simply consider each connected polygonal domain as a separate subproblem, and recurse on all of them. Let $n_i$ and $m_i$ denote the number of sites and worst-case complexity of a shortest path in subproblem $i$. This means that only reflex vertices are counted for $m_i$, which are the only relevant vertices for the spanner complexity. The sites are partitioned over the subproblems, so we have $\sum_i n_i = n$. 
The only new vertices (not of $\P$) that can be included in the subproblems are the at most three corners of the separator. Each vertex can be a reflex vertex in only one of the subproblems, thus $\sum_i m_i \leq m + 3$. In the recursion, an increase in the number of subproblems means that we might have more than $c\cdot 2^i$ vertices not of $\P$ at level $i$, but the depth of the recursion tree is then proportionally decreased. All further proofs on complexity of our spanners are written in term of $\P_\ell$ and $\P_r$, but translate to the case of multiple subproblems.

Next, we analyze the spanner construction using any 1-dimensional additively weighted $t$-spanner of size $O(n\log n)$.

\begin{lemma}\label{lem:1D_to_geodesic_spanner_holes}
The graph $\G$ is a geodesic $3t$-spanner of size $O(n\log^2 n)$.
\end{lemma}

\begin{proof}
As the 1-dimensional spanner has size $O(n\log n)$, there are still $O(n\log^2n)$ edges in~$\G$. What remains is to argue that $\G$ is a $3t$-spanner. Let $p,q$ be two sites in $S$. In contrast to the simple polygon case, a shortest path between two sites $p,q$ in $S_\ell$ (resp. $S_r$) is not necessarily contained in $P_\ell$ (resp. $P_r$). Therefore, we distinguish two different cases: either $\pi(p,q)$ is fully contained within $\P_r$ or $\P_\ell$, or there is a point $r \in \pi(p,q) \cap \lambda$ for some shortest path $\lambda$ of the separator. In the first case, there exists a path in $\G$ of length at most $3t d(p,q)$ by induction. In the second case, we have 
\begin{equation}
    d_\G(p,q) \leq d_{\G_\lambda}(p,q) \leq td_w(p_\lambda,q_\lambda) = t(d(p,p_\lambda) + d(p_\lambda,q_\lambda) + d(q,q_\lambda)).
\end{equation}
Additionally, we use that $d(p_\lambda,q_\lambda) \leq d(p_\lambda,r) + d(r,q_\lambda)$ and $d(p_\lambda,r) \leq d(p_\lambda, p) + d(p,r)$,  because of the triangle inequality, so
\begin{equation*}
    d(p_\lambda,q_\lambda) \leq d(p_\lambda,r) + d(q_\lambda,r) \leq d(p_\lambda, p) + d(p,r) + d(q_\lambda,q) + d(q,r) \leq 2 d(p,q).
\end{equation*}
It follows that $d_\G(p,q) \leq 3t d(p,q)$.
\end{proof}

Applying the simple 1-dimensional spanner of Section~\ref{fig:1D_spanner}  results in a 6-spanner. However, using the refinement of Lemma~\ref{lem:refinement}, we again obtain a $(2 + \varepsilon)$-spanner.

%%%%%%%%%%%%%%%%%%%%%%%%%%%%%%%%%%%%%%%%%%%%%%%%%%%%%%%%%%%%%%%%%%%
\subsection{Low complexity spanners in a polygonal domain}\label{sub:low_complexity_spanners_holes}

To obtain spanners of low complexity in a simple polygon, we formed groups of sites such that shortest paths within a group were disjoint from shortest paths of other groups.
We proposed two different ways of forming these groups, based on the shortest path tree of the central site $c$, and based on the shortest path tree of the separator $\lambda$. Both of these approaches do not directly lead to a low complexity spanner in a polygonal domain, as we explain next.

Both methods can still be applied in a polygonal domain, as the shortest path tree of both a site and a shortest path is still well-defined. However, it does not give us the property that we want for our groups. In particular, the second property discussed in Lemma~\ref{lem:1-dim-spanner}: each vertex of $\P$ is only used by shortest paths within $O(1)$ groups, does not hold. This is because the shortest path between two vertices $u,v$ is not necessarily homotopic to the path $\pi(u,c) \cup \pi(c,v)$. Thus paths within a group can go around a certain hole, while their shortest paths to $c$ (or $\lambda$) do not. See Figure~\ref{fig:problem_holes_spt} for an example. The construction can easily be expanded to ensure there are more sites in each group, by simply adding as many sites very close to the existing ones, or to more than three groups, by adding an additional hole above the construction with two corresponding sites. Consequently, the property that each vertex of $\P$ is used only by shortest paths within $O(1)$ groups does not hold.

\begin{figure}
    \centering
    \includegraphics{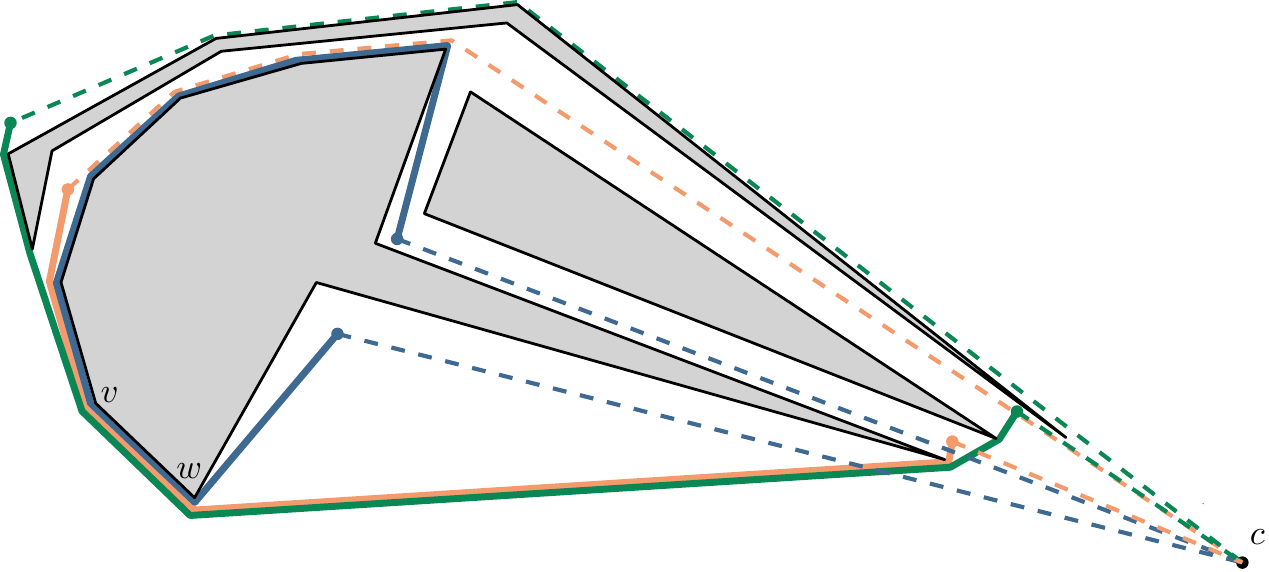}
    \caption{Assigning the sites to groups based on the shortest path tree of $c$, as described in Section~\ref{sec:4sqrt2spanner}, forms these colored groups. The shortest path from each site to $c$ is shown dashed. Each shortest path between two sites of a group contains both vertices $v$ and $w$.}
    \label{fig:problem_holes_spt}
\end{figure}

So far, we assumed that every edge $(p,q) \in E$ is a shortest path between $p$ and $q$. To obtain a spanner of low complexity, we can also allow an edge between $p$ and $q$ to be any path between the two sites. We call such a spanner a relaxed geodesic spanner. Note that our lower bounds still hold in this case. In the lower bound for a $(3-\varepsilon)$-spanner (Figure~\ref{fig:lowerbound_3spanner}), every path between $p \in S_\ell$ and $q \in S_r$ has complexity $\Theta(m)$, and in the general lower bound (Figure~\ref{fig:lower_bound_general}) we can easily adapt the top side of the polygon such that any path between two sites $p,q$ has the same complexity as $\pi(p,q)$.

\subsubsection{A $12$-spanner of complexity $O(m\sqrt{n} + n\log^2n)$}\label{sub:12-spanner}

To obtain a low complexity spanner in a polygonal domain, we adapt our techniques for the simple-polygon $4\sqrt{2}$-spanner in such a way that we avoid the problems we just sketched. The main difference with the simple polygon approach is that for an edge $(p_\lambda,q_\lambda)$ in the 1-dimensional spanner, the edge $(p,q)$ that we add to $\G$ is no longer $\pi(p,q)$. Instead, let $(p,q)$ be the shortest path from $p$ to $q$ via $p_\lambda$ and $q_\lambda$, excluding any overlap of the path. We denote this path by $\pi_\lambda(p,q)$. This path is not unique, for example when $\pi(p,p_\lambda)$ is not unique, but choosing the path that is used to construct $\mathit{SPT}_\lambda$ will do. Formally, $\pi_\lambda(p,q)$ is defined as follows. 
\begin{definition}\label{def:pi_lambda}
    The path $\pi_\lambda(p,q)$ is given by:
    \begin{itemize}
        \item $\pi(p,p_\lambda) \cup \pi(p_\lambda, q_\lambda) \cup \pi(q_\lambda,q)$, where $ \pi(p_\lambda, q_\lambda) \subseteq \lambda$, if $\pi(p,p_\lambda)$ and $\pi(q,q_\lambda)$ are disjoint,
        \item $\pi(p,r) \cup \pi(r,q)$, where $r$ denotes the closest point to $p$ of $\pi(p,p_\lambda) \cap \pi(q,q_\lambda)$, otherwise.
    \end{itemize}
\end{definition}    

One of the properties that we require of the groups, see Lemma~\ref{lem:group_properties}, has changed, namely that each vertex of $\P$ is only used by shortest paths within $O(1)$ groups. Instead of the shortest paths between sites in a group, we consider the paths $\pi_\lambda(p,q)$ of Definition~\ref{def:pi_lambda}. The following lemma shows that we can obtain a spanner with similar complexity as in a simple polygon when groups adhere to this adjusted property.

\begin{lemma}\label{lem:group_properties_polygonal_domain}
    If the groups adhere to the following properties, then $\G$ has complexity $O(m\sqrt{n}+n\log^2n)$:
    \begin{enumerate}
        \item each group contains $\Theta(\sqrt{n})$ sites, and
        \item each vertex of $\P$ is used by paths $\pi_\lambda$ within $O(1)$ groups.
    \end{enumerate}
\end{lemma}
\begin{proof}
    Note that the complexity of any path $\pi_\lambda(p,q)$ is $O(m)$, as it can use a vertex of $\P$ at most once. 
    Thus the proof of Lemma~\ref{lem:group_properties} directly implies that the complexity of the edges in one level of the 1-dimensional spanner is $O(m\sqrt{n} + n)$.

    In the 1-dimensional recursion, splitting the sites by $O$ no longer corresponds to a horizontal split in the polygon. However, the paths $\pi(p, p_\lambda)$, $p \in S_\ell$, are still disjoint from the paths $\pi(q,q_\lambda)$, $q \in S_r$. For the two subproblems generated by the split by $O$ it thus still holds that $m_1 + m_2 = m$, where $m_i$ denotes the maximum complexity of a path in subproblem $i$. Lemma~\ref{lem:1_dim_recursion} states that this recursion solves to $O(m\sqrt{n} + n\log n)$.

    In the recursion where the domain is partitioned into two subpolygons $\P_\ell$ and $\P_r$, we now add at most three new vertices to the polygonal domain, namely the three corners of the sp-separator. Each vertex of $\P$ can only be a reflex vertex in either $\P_\ell$ or $\P_r$, so $m_1 + m_2 \leq m + 3$. Lemma~\ref{lem:1_dim_recursion} implies that this recursion for the complexity solves to $O(m\sqrt{n} + n\log^2n)$.
\end{proof}

As in Lemma~\ref{lem:1-dim-spanner}, we form the groups based on the traversal of the shortest path tree $\mathit{SPT}_\lambda$. We again include all sites in $S$ in the shortest path tree. Whenever a node $v$ of $\lambda$ has multiple children, we let all nodes that correspond to vertices/sites that lie to the \emph{left} of $\lambda$ come before nodes that correspond to vertices/sites that lie to the right of~$\lambda$ in the in-order traversal. Within these sets, the vertices/sites are ordered from bottom to top, as seen from $v$. See Figure~\ref{fig:groups_with_holes} for an example. The subtree rooted at the start and end point of~$\lambda$ is simply a part of the shortest path tree of the start/end point. The first $\lceil \sqrt{n}\rceil$ sites in the in-order traversal are in $S_1$, the second $\lceil \sqrt{n} \rceil$ in $S_2$, etc.

\begin{figure}
    \centering
    \includegraphics{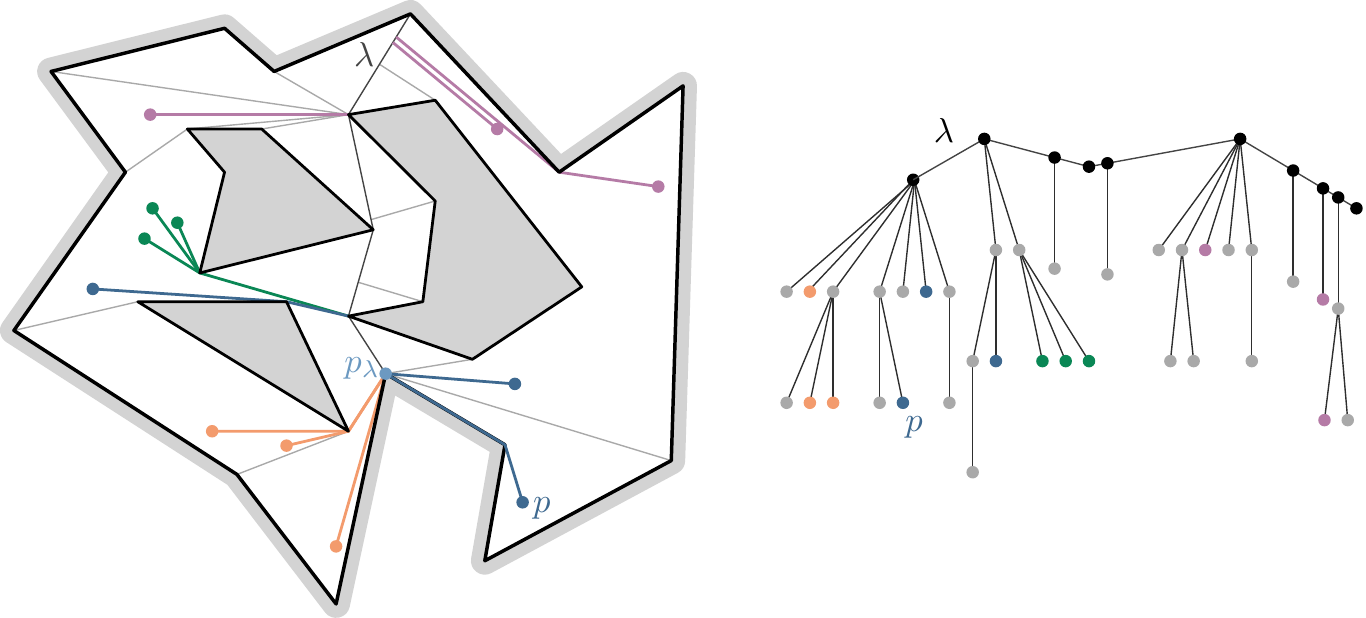}
    \caption{The shortest path tree $\mathit{SPT}_\lambda$ and the corresponding group assignment.}
    \label{fig:groups_with_holes}
\end{figure}

Clearly, these groups adhere to property 1 of Lemma~\ref{lem:group_properties_polygonal_domain}.
To show these groups adhere to property 2, we again consider for each group $S_i$ the minimal subtree $\T_i$ of $\mathit{SPT}_\lambda$ that contains all $p \in S_i$. Whenever there is more than one vertex of $\lambda$ in $\T_i$, we choose the leftmost of these vertices as the root of $\T_i$.

\begin{lemma}\label{lem:path_in_Ti}
An edge $\pi_\lambda(p,q)$ with $p_\lambda, q_\lambda \in S_i$ bends only at vertices in $\T_i$.
\end{lemma}
\begin{proof}
The path $\pi_\lambda(p,q)$ is a subpath of $\pi(p,p_\lambda) \cap \pi(p_\lambda,q_\lambda) \cap \pi(q,q_\lambda)$. In particular, it is the path from $p$ to $q$ in $\T_i$ via their lowest common ancestor. This is $p_\lambda$ or $q_\lambda$ when $\pi(p,p_\lambda) \cap \pi(q,q_\lambda) = \emptyset$ and the vertex $r$, as in Definition~\ref{def:pi_lambda}, otherwise.
\end{proof}

\begin{lemma}
Any vertex in $\mathit{SPT}_\lambda$ occurs in at most two trees $\T_i$ and $\T_j$ as a  non-root node.
\end{lemma}
\begin{proof}
Follows directly from the proof of  Lemma~\ref{lem:occurence_spt}.
\end{proof}

Except for the root nodes, property 2 thus holds. As each $\T_i$ has only one root node by definition, the number of groups using a single root node $v$ in their paths may be large, but the sum of the number of groups that use a node over all root nodes is still $O(\sqrt{n})$. From this and Lemma~\ref{lem:group_properties_polygonal_domain} we conclude that $\G$ has complexity $O(m\sqrt{n} + n \log^2n)$.

\begin{lemma}
The graph $\G$ is a geodesic $12$-spanner of size $O(n \log^2n)$.
\end{lemma}
\begin{proof}
The number of edges is exactly the same as in the $4\sqrt{2}$-spanner in a simple polygon.
The way the groups are formed in the construction of the 1-dimensional spanner $\G_\lambda$ does not influence its spanning ratio, thus $\G_\lambda$ is a 4-spanner (see Lemma~\ref{lem:4sqrt2-spanner}). Note that even for our redefined edges in $\G$ it holds that $d_\G(p,q) \leq d_{\G_\lambda}(p_\lambda,q_\lambda)$. Lemma~\ref{lem:1D_to_geodesic_spanner_holes} then directly implies that $\G$ is a 12-spanner.
\end{proof}

\subsubsection{A $6k$-spanner of complexity $O(mn^{1/k} + n \log^2n)$}
The generalization of the construction, discussed in Section~\ref{sub:general_complexity_spanner}, where $\Theta(n^{1/k})$ groups are recursively partitioned into smaller groups, is also applicable in a polygonal domain. As our groups adhere to the required properties, the complexity of this spanner remains $O(mn^{1/k} + n\log^2 n)$, as in the simple polygon, but the spanning ratio increases to $6k$. 

\begin{lemma}\label{lem:6k-spanner-holes}
Let $S$ be a set of $n$ point sites in a polygonal domain $\P$ with $m$ vertices, and let $k \geq 1$ be any integer constant. There exists a geodesic $6k$-spanner of size~$O(n\log^2n)$ and complexity $O( m n^{1/k} + n\log^2 n)$. 
\end{lemma}

%%%%%%%%%%%%%%%%%%%%%%%%%%%%%%%%%%%%%%%%%%%%%%%%%%%%%%%%%%%%%%%%%%%
\subsection{Construction algorithm}

In this section we discuss an algorithm to compute the geodesic spanners of Section~\ref{sub:low_complexity_spanners_holes}. The following gives an overview of the algorithm that computes a $6k$-spanner of complexity $O(m n^{1/k} + n\log^2 n)$ in $O(n^2 \log m + nm\log m)$ time.

\begin{enumerate}
        \item Find an sp-separator such that $\P$ is partitioned into two polygons $\P_\ell$ and $\P_r$, and $S_\ell$ contains at least $2n/9$ and at most $2n/3$ sites using the algorithm of Theorem \ref{thm:find_separator_holes_improved}.
        
        \item For each shortest path $\lambda$ of the separator:
        \begin{enumerate}
            \item For each $p \in S$ find the weighted point $p_\lambda$  on $\lambda$ and add this point to $S_\lambda$ using the algorithm of Lemma~\ref{lem:find_projections_holes}. 
            \item Compute an additively weighted 1-dimensional spanner $\G_\lambda$ on the set $S_\lambda$. 
            \item For every edge $(p_\lambda, q_\lambda) \in E_\lambda$ add the edge $(p,q) = \pi_\lambda(p,q)$ to $\G$. 
        \end{enumerate}
        \item Recursively compute spanners for $S_\ell$ in $\P_\ell$ and $S_r$ in $\P_r$.
    %\end{enumerate}
\end{enumerate}

The algorithm starts by finding a balanced sp-separator. According to Theorem~\ref{thm:find_separator_holes_improved} this takes $O(n \log n + m\log m)$ time. We then continue by building a 1-dimensional additively weighted spanner on each of the shortest paths defining $\lambda$ as follows.

\begin{lemma}\label{lem:find_projections_holes}
    We can compute the closest point $p_\lambda$ on $\lambda$ and $d(p,p_\lambda)$ for all sites $p \in S$, and the shortest path tree $\mathit{SPT}_\lambda$, in $O((m+n)\log m)$ time.
\end{lemma}
\begin{proof}
Hershberger and Suri~\cite{SPMHershbergerSuri} show how to build the shortest path map of a point site in a polygonal domain in $O(m\log m)$ time. They also note that this extends to non-point sources, such as line segments, and to $O(m)$ sources, without increasing the running time. We can thus build the shortest path map of $\lambda$ in $\P$ in $O(m\log m)$ time, using that $\lambda$ has complexity $O(m)$.  After essentially triangulating each region of the shortest path map using the vertices of $P$, we obtain $\mathit{SPT}_\lambda$ in the same time bound. After building a point location data structure for this augmented shortest path map, we can query it for each site $p \in S$ to find $p_\lambda$ and $d(p,p_\lambda)$ in $O(\log m)$ time.
\end{proof}

\begin{lemma}\label{lem:1-dim-spanner-holes}
    Given $\mathit{SPT}_\lambda$, we can construct a $4$-spanner $\G_\lambda$ on the additively weighted points $S_\lambda$, where the groups adhere to the properties of Lemma~\ref{lem:group_properties_polygonal_domain}, in $O(n\log n + m)$ time.
\end{lemma}
\begin{proof}
    As in Lemma~\ref{lem:1-dim-spanner}, we can reuse $\mathit{SPT}_\lambda$ to form the groups based on the ordering produced by an in-order traversal of the tree. See Section~\ref{sub:12-spanner} for an exact description of this ordering. The ordering allows us to form the groups for a level of the 1-dimensional spanner in $O(n)$ time, thus the total running time is $O(n\log n)$.
\end{proof}

For the general $2k$-spanner for additively weighted sites, using the ordering of the sites and Lemma~\ref{lem:1-dim-spanner-bottom-up} to build the tree of groups bottom up implies that we can construct this spanner in $O(n\log n + m)$ time as well.

After computing the additively weighted spanner $\G_\lambda$, we add edge $(p,q) = \pi_\lambda(p,q)$ to $\G$ for every edge $(p_\lambda,q_\lambda) \in \G_\lambda$. We can either compute and store these edges explicitly, which would take time and space equal to the complexity of all added edges, or we can store them implicitly by only storing the points $p_\lambda$ and $q_\lambda$, or the point $r$ from Definition~\ref{def:pi_lambda} when the paths are not disjoint. The point $r$ is the lowest common ancestor of the nodes $p$ and $q$ in $\mathit{SPT}_\lambda$. This can be computed in $O(1)$ time after $O((n+m)\log(n+m))$ preprocessing time~\cite{LCA_queries}.

As there are at most three shortest path that define the separator, step 2 takes $O((n+m)\log(n+m))$ time in total. This means that this is also the dominant term in the construction. The total running time is thus $O((n+m)\log^2n \log m)$.

\begin{restatable}{theorem}{spannerHolesTime}\label{thm:6k-spanner-holes-time}
Let $S$ be a set of $n$ point sites in a polygonal domain $\P$ with $m$ vertices, and let $k \geq 1$ be any integer constant. We can build a relaxed geodesic $6k$-spanner of size~$O(n\log^2n)$ and complexity $O(m n^{1/k} + n\log^2 n)$ in $O((n+m)\log^2n \log m + K)$ time, where $K$ is the output complexity. 
\end{restatable}

\input{./simple_polygon_eps}

%%%%%%%%%%%%%%%%%%%%%%%%%%%%%%%%%%%%%%%%%%%%%%%%%%%%%%%%%%%
\section{Lower bounds for complexity}\label{sec:lower_bounds}

In this section, we consider lower bounds on the complexity of spanners. We first describe a simple $\Omega(nm)$ lower bound construction for a (relaxed) geodesic $(3-\eps)$-spanner, and then prove a (slightly worse) $\Omega(mn^{1/(t-1)})$ lower bound construction for a $(t-\eps)$-spanner.

\subsection{Lower bound for $(3-\varepsilon)$-spanners}\label{sub:lower_bound_3-eps}

\begin{theorem}
For any constant $\varepsilon \in (0,1)$, there exists a set of $n$ point sites in a simple polygon $P$ with $m = \Omega(n)$ vertices for which any (relaxed) geodesic $(3-\varepsilon)$-spanner has complexity $\Omega(mn)$. 
\end{theorem}

\begin{proof}
Consider the construction given in Figure~\ref{fig:lowerbound_3spanner}. We assume that $m = \Omega(n)$. We split the sites into two sets $S_\ell$ and $S_r$ equally. The sites lie in long `spikes' of length $\ell$, either on the left ($S_\ell$) or right ($S_r$) of a central passage of complexity $\Theta(m)$. We show that this construction gives a complexity of $\Omega(mn)$ for any $(3-\varepsilon)$-spanner.

When $h$ gets close to 0, the distance between any two sites $p,q$ approaches $2\ell$. To get a $(3-\varepsilon)$-spanner, we can thus have at most one intermediate site on the path from $p$ to $q$. We assume that all possible, constant complexity, edges between vertices on the same side of the construction are present in the spanner. To make sure sites on the left also have (short) paths to sites on the right, we have to add some additional edges that go through the central passage of the polygon, which forces paths to have complexity $\Theta(m)$. We will show that we need $\Theta(n)$ of these edges, each of complexity $\Theta(m)$, to achieve a $(3-\varepsilon)$-spanner, thus proving the $\Omega(nm)$ lower bound.
  
  Let $q \in S_r$. For each $p \in S_\ell$, we need a path with at most one intermediate site to $q$. There are two ways to achieve this: we can go through an intermediate site on the left, or on the right. In the first case, we add an edge from a site $p'\in S_\ell$ to $q$. In the second case, we need to add an edge from each $p \in S_\ell$ to any site $q'\in S_r$. In the first case we thus add only one edge, while in the second case we add $\Theta(n)$ edges (that go through the central passage). Let $k \geq 1$ be the number of sites in $S_r$ that have a direct edge to some site in $S_\ell$. If $k < |S_r|$, then there is some site $q \in S_r$ for which we are in case two, and we thus have $\Theta(n)$ edges of complexity $\Theta(m)$. If $k = |S_r|$, then there is a direct edge to each of the $\Theta(n)$ sites in $S_r$, and we therefore also end up with a complexity of $\Omega(nm)$.
  \end{proof}

\subsection{General lower bounds}\label{sub:general_lower_bound}
\lowerBound*
\begin{proof}
Consider the construction of the polygon $P$ shown in Figure~\ref{fig:lower_bound_general}. Let $p_1,...,p_n$ be the sites from left to right. Thus, the complexity of any path from the $i$-th site $p_i$ to the $j$-th site $p_j$ is at least $|i-j|$. When $h$ is close to $0$, the distance between any two sites approaches $2\ell$. To achieve a spanning ratio of $(t-\varepsilon)$, the path in the spanner from $p_i$ to $p_j$ can visit at most $t -2$ other vertices. In other words, we can go from $p_i$ to $p_j$ in at most $t-1$ hops. This is also called the hop-diameter of the spanner.

As the spanning ratio is determined only by the number of hops on the path, we can model the spanner in a much simpler metric space $\vartheta_n$. This is a 1-dimensional Euclidean space with $n$ points $v_1,...,v_n$ that lie on the $x$-axis at coordinates $1,2,...,n$. The edge $(v_i,v_j)$ thus has length (or weight) $|i-j|$. Any spanning subgraph of $\vartheta_n$ of hop-diameter $h$ and total weight $w$ (the weight of a graph is the sum of the weights of its edges) is in one-to-one correspondence to an $(h + 1 -\varepsilon)$-spanner of $P$ of complexity $\Theta(w)$. Denitz, Elkin, and Solomon~\cite{shallow_low_light_trees} prove the following on the relation between the \textit{hop-radius} and weight of any spanning subgraph of $\vartheta_n$. The hop-radius $h(G,r)$ of a graph $G$ with respect to a root $r$ is defined as maximum number of hops that is needed to reach any vertex in the graph from the root. The hop-radius $h(G)$ of $G$ is then defined as $\min_{r\in V} h(G,r)$. Note that the hop-diameter is an upper bound on $h(G)$. 
\begin{lemma}[Dinitz \etal~\cite{shallow_low_light_trees}] \label{lem:lower_bound_vartheta}
For any sufficiently large integer $n$ and positive integer $h < \log n$, any spanning subgraph of $\vartheta_n$ with hop-radius at most $h$ has weight at least $\Omega(h \cdot n^{1+1/h})$.
\end{lemma}
The lemma implies that any $(t-\varepsilon)$-spanner of $P$, which has hop-diameter $t-1$, has complexity $\Omega((t-1) \cdot n^{1+1/(t-1)}) = \Omega(n^{1+ 1/(t-1)})$, for constant $t$.

\begin{figure}
    \centering
    \includegraphics{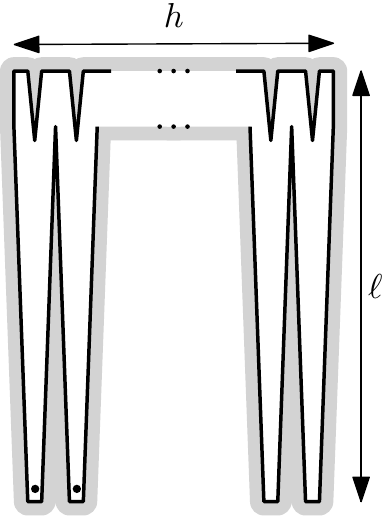}
    \caption{Construction of a simple polygon $P$ for which any $(t-\varepsilon)$-spanner has complexity $\Omega(n^{1 + 1/(t-1)})$. The shortest path from the $i$-th to the $j$-th site has complexity $|i-j|$.}
    \label{fig:lower_bound_general}
\end{figure}

To achieve a lower bound for $m> n$, we slightly adapt our polygon such that the path between two adjacent spikes has complexity $\Theta(m/n)$. This implies that the path between $p_i$ and $p_j$ has complexity $m/n \cdot |i-j|$ instead. For this adapted polygon $P'$, any spanning subgraph of $\vartheta_n$ of hop-diameter $h$ and total weight $w$ is in one-to-one correspondence to an $(h + 1 -\varepsilon)$-spanner of $P'$ of complexity $\Theta(m/n \cdot w)$. It follows from Lemma~\ref{lem:lower_bound_vartheta} that any $(t-\varepsilon)$-spanner of $P'$ has complexity $\Omega(m/n \cdot n^{1+ 1/(t-1)}) = \Omega(mn^{1/(t-1)})$.
\end{proof}
\begin{corollary}
For any integer constant $t \geq 2$, there exists a set of $n$ point sites in a simple polygon $P$ with $m = \Omega(n)$ vertices for which any (relaxed) geodesic $t$-spanner has complexity~$\Omega(mn^{1/t})$.
\end{corollary}

\section{Conclusion}
In this paper, we have taken a first look at a new measure of compactness of geodesic spanners, namely the spanner complexity. We presented both constructions for low complexity spanners, as well as lower bounds on the complexity for a given spanning ratio. A clear direction for future work lies in closing the gap between the upper and lower bounds. This gap is much larger for a polygonal domain, as our spanner of complexity $O(mn^{1/k} + n \log^2n)$ has a spanning ratio of $6k$, while the simple polygon spanner of similar complexity has a spanning ratio of only $2k + \eps$. This difference lies in the application of the refinement suggested by Abam, de Berg, and Seraji~\cite{SpannerPolyhedralTerrain}. For a simple polygon, we managed to adapt their approach to reduce the spanning ratio such that the complexity of the spanner is increased by only a constant factor. However, it does not seem straightforward to apply the same refinement in a polygonal domain while retaining the low complexity property. 

For a polygonal domain, Abam~\etal~\cite{SpannerPolygonalDomain} suggest a different construction from our approach that results in a spanner of size dependent on the number of holes in $P$. Our techniques to construct a low complexity spanner can also be applied to this construction. This results in a $6k$-spanner of size~$O(\sqrt{h}n\log^2n)$ and complexity $O(\sqrt{h} ( m n^{1/k} + n\log^2 n))$. As the separators used here are always line segments instead of general shortest paths, it would be interesting to see if it is possibly easier to apply the refinement of~\cite{SpannerPolyhedralTerrain} in this setting.

\bibliography{bibliography}

\newpage
\appendix

\section{Alternative proof for the existence of a balanced separator}\label{ap:separator}

In an earlier version of this paper~\cite{complexity_spanners}, we
gave a constructive proof for the existence of a balanced sp-separator
using a similar approach to Abam, de Berg, and
Seraji~\cite{SpannerPolyhedralTerrain}. As we consider a polygonal
domain instead of a polyhedral terrain, we need to adapt the approach
slightly to fit this setting. Moreover, we correct a technical issue
of~\cite{SpannerPolyhedralTerrain}. In this section we first give our
proof for the existence of a separator in a polygonal domain, and then
comment how to fix the technical issue
in~\cite{SpannerPolyhedralTerrain} using our approach.

\subparagraph{Existence of a separator.} Let $P$ be a polygonal domain
with $h$ holes. We start by trying to find a 1-separator from an arbitrary fixed point $u \in \partial \P$ to a point $v \in \partial \P$. The point $v$ is moved clockwise along the boundary of $\P$, starting at $u$, to find a separator satisfying our constraints, i.e. $\P_\ell$ contains between $2n/9$ and $2n/3$ sites of $S$. This way, we either find a balanced 1-separator, or jump over at least $2n/3 - 2n/9 = 4n/9$ sites at a point $v$. In this case, the region bounded by two shortest paths between $u$ and $v$ contains at least $4n/9$ sites. We then try to find a 2- or 3-separator contained within this region.

To find a 2- or 3-separator, we construct a sequence of 3-separators $\Delta_0 \supset \Delta_1 \supset ... \supset \Delta_k$, where either the final 3-separator $\Delta_k$ contains between $2n/9$ and $2n/3$ sites, or we find a balanced 2-separator within $\Delta_k$. During the construction, the invariant that $|\Delta_i \cap S| \geq 2n/9$ is maintained. The first 3-separator, $\Delta_0$, has as corners the two points $u$ and $v$ on $\partial \P$ from before, and a point $w$, which is an arbitrary point on one of the two shortest paths connecting $u$ and $v$. So, $\Delta_0$ is exactly the region bounded by the two shortest path connecting $u$ and $v$ that contains at least $4n/9$ sites.

Whenever $\Delta_i$ contains at most $2n/3$ sites, we are done. If not, we find $\Delta_{i+1}$ as follows. If $\Delta_i$ is a degenerate 3-separator, we simply select a subpath of the shortest path that contains $2n/3$ sites. Similarly, if at least $2n/9$ sites lie on any of the bounding shortest paths. If neither of these cases holds, then there are at least $2n/3 + 1 - 3(2n/9 -1) = 4$ sites in the interior of $\Delta_i$. In the following, we find a 3-separator $\Delta_{i+1}$ with either $|\Delta_{i+1} \cap S| < |\Delta_i \cap S|$, or $|\Int(\Delta_{i+1}) \cap S| < |\Int(\Delta_i) \cap S|$, or a valid 2-separator contained within $\Delta_i$. Because each $\Delta_i$ contains either fewer sites, or fewer sites in its interior, than its predecessor, we eventually find a 3-separator with the desired number of sites, or we end up in one of the easy degenerate cases. In the following description, we drop the subscript $i$ for ease of description.

We define a \emph{good} path to be a path from a point $p \in \Delta$ to a corner $u$ to be a shortest path $\pi(p,u)$ that is fully contained within $\Delta$. To make sure our definition is also correct when one of the corners lies inside $\Delta$, we do not allow the path to cross $\pi(v,w)$. See Figure~\ref{fig:good_path} for an illustration. Essentially, we see the coinciding part of the shortest paths as having an infinitesimally small separation between them. The following provides a formal definition of a good path.

\begin{definition}\label{def:good_path}
A shortest path $\pi(p,u)$ from a point $p$ that lies in a 3-separator $\Delta$ to a corner $u$ of $\Delta$, is a \emph{good path} if it is fully contained within $\Delta$, and $\pi(p,u)$ is a shortest path in the polygonal domain $\P \cap \Delta$, where the outer polygon is $\Delta$.
\end{definition}

\begin{figure}
    \centering
    \includegraphics{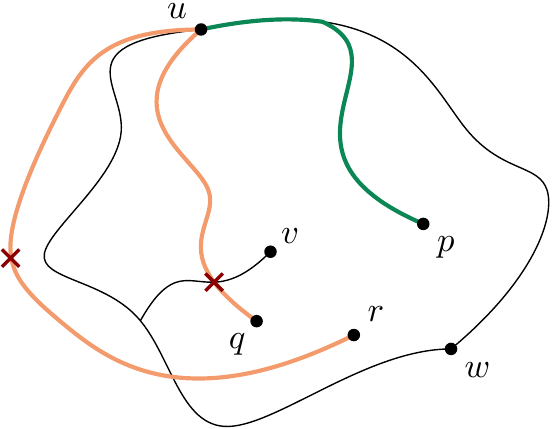}
    \caption{Of the shortest paths $\pi(p,u)$, $\pi(q,u)$, and $\pi(r,u)$ only $\pi(p,u)$ is a good path.}
    \label{fig:good_path}
\end{figure}

We consider the closed region $Z_u \subset \Delta$ such that for every $p \in Z_u$ there is a good path to $u$. Because for any $p'$ on a good path $\pi(p,u)$ the path $\pi(p',u)$ is also a good path, $Z_u$ is connected. $Z_u$ is bounded by $\pi(u,v)$, $\pi(u,w)$, $\partial H$, and a curve $B_u$ that connects $v$ to $w$, see Figure~\ref{fig:Good_path_v_w}. Because we do not consider $\partial H$ to be part of $B_u$, this is a possibly disconnected curve that consists of edges from the shortest path map of $u$ and $\Delta$. The shortest path map of a point $p$ in a polygonal domain $\P$ partitions the free space into maximal regions, such that for any two points in the same region the shortest paths from $u$ to both points use the same vertices of $\P$~\cite{NewSPMAlgo}. We call the curves of the shortest path map for which there are two topologically distinct paths from $p$ to any point on the curve \emph{walls}. We prove the following lemma of Abam \etal~\cite{SpannerPolyhedralTerrain} for our definition of a good path. 

\begin{lemma}[Lemma 3.2 of~\cite{SpannerPolyhedralTerrain}]
For any point $z \in B_u$, there are good paths $\pi(z, u)$, $\pi (z, v)$, and $\pi (z, w)$ to the three corners of $\Delta$.
\end{lemma}

\begin{proof}
\begin{figure}
    \centering
    \includegraphics{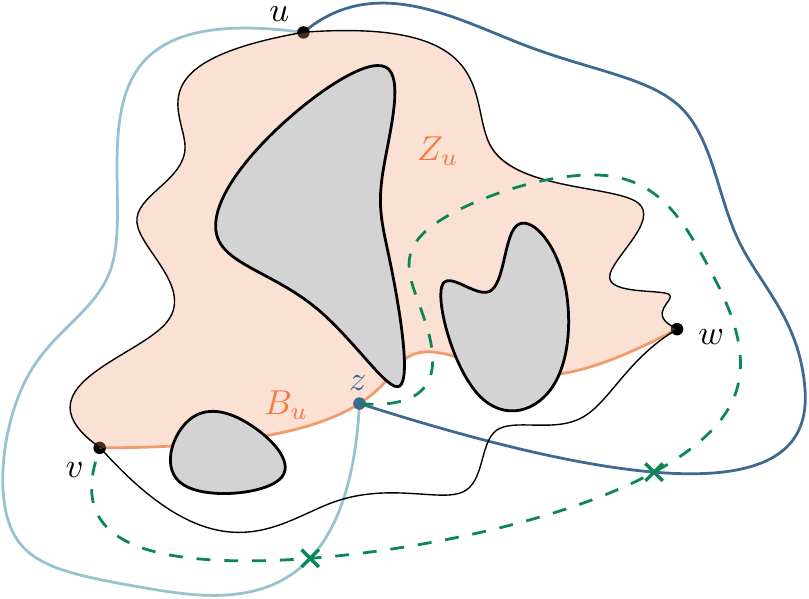}
    \caption{Any path from $z$ to $v$ that exits $\Delta$ through $\pi(u,w)$ must intersect a path from $z$ to $u$.}
    \label{fig:Good_path_v_w}
\end{figure}
By definition of $B_u$, there is a good path from $z$ to $u$. When $z$ lies on $\pi(v,w)$, the subpaths of $\pi(v,w)$ from $z$ to $v$ and to $w$ are good paths. So, assume $z \notin \pi(v,w)$. We will prove by contradiction that there is a good path from $z$ to $v$, and by symmetry from $z$ to $w$. Suppose there is no good path from $z$ to $v$. Because $z$ is on $B_u$, there is also a shortest path from $z$ to $u$ that is not a good path, so it is not contained within $\Delta$, or it crosses $\Delta$. This path $\pi(z,u)$ must exit (or cross) $\Delta$ through $\pi(v,w)$, because otherwise the path could simply continue along $\pi(v,u)$ or $\pi(w,u)$ and stay within $\Delta$. Similarly, $\pi(z,v)$ must exit $\Delta$ through $\pi(u,w)$. The path $\pi(z,u)$ either goes around $v$ or $w$, as shown in Figure~\ref{fig:Good_path_v_w}. In both cases, any path to $v$ that starts at $z$ and exits (or crosses) through $\pi(u,w)$ and does not intersect $\pi(u,w)$ again, must intersect the path $\pi(z,u)$. As these shortest paths start at the same point, this is a contradiction with the fact that two shortest paths can only cross once.
\end{proof}

\begin{figure}
    \centering
    \includegraphics{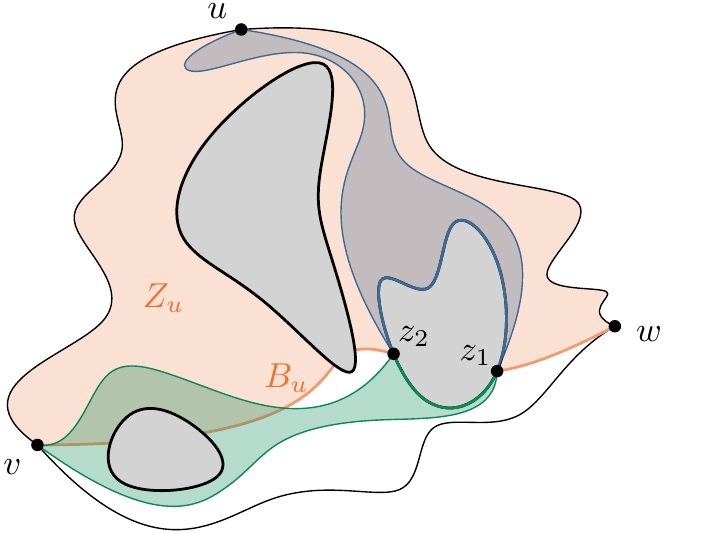}
    \caption{The region $Z_u$ where all points have good paths to $u$. In blue and green the 2-separator with corners $z_1,u,z_2$ and $z_1,v,z_2$, respectively.}
    \label{fig:Zu_region}
\end{figure}

This lemma implies that for any point $z$ on $B_u$, the 3-separator with corners $u$, $v$, and $z$ is contained within $\Delta$. We use this observation by moving a point $z$ along $B_u$ from $w$ to $v$. At the start, the 3-separator defined by $u$, $v$, and $z$, which we denote by $T_z$, is equal to $\Delta$. At the end, $T_z$ is equal to the degenerate 3-separator $\pi(u,v)$. Note that, in contrast to the situation of Abam \etal~\cite{SpannerPolyhedralTerrain} for a terrain, $B_u$ is not necessarily continuous, as it can be interrupted by holes. But, if $B_u$ intersects a hole, it intersects this hole exactly twice, because $Z_u$ is connected. A directed walk along $B_u$, jumping at holes, is thus still well-defined. We walk a point $z$ along $B_u$ until one of the following happens: (1) $|S \cap T_z|$ or $|S \cap \Int(T_z)|$ decreases, or (2) $z$ encounters a hole.

In case (1), either $T_z$ contains at least $2n/9$ sites, then we set $\Delta_{i+1} := T_{z}$, or we jump over at least $2n/3-2n/9 = 4n/9$ sites. This can happen because the shortest path to $u$, to $v$, or both jump over a hole. We assume the path to $u$ jumps, the approach for when the path to $v$ jumps is symmetric. The 3-separator $u,z,w'$, with $w'$ on one of the two $\pi(z,u)$, could contain the same number of sites, in the closure and in its interior, as $\Delta$. Therefore, we select an arbitrary site $s \in S$ that lies in the region bounded by the two shortest paths. We then consider the two 3-separators with $u$, $z$, and $s$ as corners. As $s$ now lies on the boundary of the 3-separators, both of these separators contain less sites in their interior than $\Delta$. And, as they partition a region that contains at least $4n/9$ sites, one of them contains at least $2n/9$ sites. We set $\Delta_{i+1}$ to be this 3-separator.

In case (2), let $z_1$ denote the point we encounter the hole $H_i$, and $z_2$ be the point where $B_u$ exits $H_i$, see Figure~\ref{fig:Zu_region}. Whenever $|T_{z_2} \cap S| \geq 2n/9$, we simply continue the walk at $z_2$, until we again end up in one of the two cases. If not, then there are at least $2n/3 - 2n/9 = 4n/9$ sites in $T_{z_1} \setminus T_{z_2}$, because $T_{z_2} \subset T_{z_1}$. Note that $T_{z_1} \setminus T_{z_2}$ is exactly the union of the two 2-separators  $z_1,u,z_2$ and $z_1,v,z_2$. This means that either the 2-separator $z_1,u,z_2$ or $z_1,v,z_2$ contains at least $2n/9$ sites. Suppose that the separator including $u$ contains at least $2n/9$ sites. We then try to find a balanced 2-separator $z_1,u,z^*$ with $z^* \in \partial H_i$ that is contained within the 2-separator $z_1,v,z_2$. For each point $z^* \in \partial H_i$ on the ``$u$-side'' of $H_i$ (blue in Figure~\ref{fig:Zu_region}), $\pi(z^*,u)$ is contained within the 2-separator $z_1,u,z_2$, as it cannot cross $\pi(z_1,u)$ and $\pi(z_2,u)$. As we did to find a 1-separator, we walk $z^*$ along $\partial H_i$ from $z_1$ to $z_2$ on the ``$u$-side''. As before, we now find a balanced 2-separator, and we are done, or we jump over at least $2n/3 - 2n/9 = 4n/9$ sites. Again, this region does not necessarily contain less sites than $\Delta$, thus we continue as before by selecting a site $s$ in this region as the third corner of $\Delta_{i+1}$. Similarly, when the 2-separator with $v$ as a corner contains more than $2n/9$ sites, we walk $z^*$ along the other side of $H_i$ and consider the 2-separator $z_1,v,z^*$.

\subparagraph{Correction of a technical issue in~\cite{SpannerPolyhedralTerrain}.}
Abam, de Berg, and Seraji~\cite{SpannerPolyhedralTerrain} present an approach to find a balanced sp-separator on a terrain $\T$. They define an sp-separator as either a 1-separator or a 3-separator, where the three shortest paths are disjoint except for their mutual endpoints. However, on a polyhedral terrain these paths might also not be disjoint. This can, for example, happen when choosing a site $s$ in the area bounded by two shortest paths as a new corner.  Figure~\ref{fig:not_disjoint_paths} illustrates this for a polygonal domain, and can be generalized to a terrain by making the holes into very high walls. Their subsequent definition of a good path, which is simply a shortest path contained within $\Delta$, would then imply that the entire interior of the box, including the blue region in Figure~\ref{fig:not_disjoint_paths_Zu}, is contained in $Z_u$. The question is then how we move $z$ from $w$ to $v$. If we would move $z$ along $\pi(w,v)$ at the start, the site $w$ will immediately enter the interior of the 3-separator $u,v,z$, see Figure~\ref{fig:not_disjoint_paths_Tz}. The number of sites in the interior has thus actually \emph{increased} instead of decreased. In a polyhedral terrain, we can define a good path similar to Definition~\ref{def:good_path}, but then consider the path $\pi(p,u)$ in the polyhedral terrain $\T \cap \Delta$. Using our renewed definition, the proof of Abam \etal~\cite{SpannerPolyhedralTerrain} also holds in the case of coinciding edges on a terrain.

\begin{figure}
    \centering
    \begin{subfigure}[b]{0.47\textwidth}
        \includegraphics[page=1]{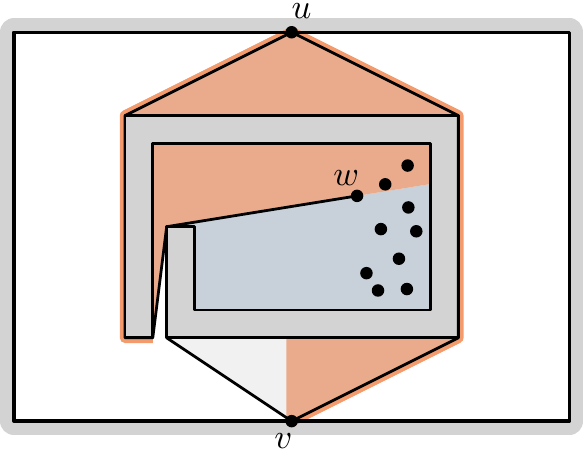}
    \caption{The region $Z_u$ according to our definition of a good path in orange. The region $Z_u$ according to~\cite{SpannerPolyhedralTerrain} includes the blue region.}
    \label{fig:not_disjoint_paths_Zu}
    \end{subfigure}
    \hspace{0.04\textwidth}
    \begin{subfigure}[b]{0.47\textwidth}
        \includegraphics[page=2]{Problem_Abam_disjoint_paths}
    \caption{The 3-separator $T_z$ has more sites in its interior than the original 3-separator $\Delta$.\\}
    \label{fig:not_disjoint_paths_Tz}
    \end{subfigure}
    \caption{A polygonal domain where the shortest paths in the separator $\Delta$ with corners $u,v,w$ are not disjoint.}
    \label{fig:not_disjoint_paths}
\end{figure}
\end{document}

%% file: simple_polygon_eps.tex
\section{Improving the spanning ratio in a simple polygon}\label{sec:simple_polygon_eps}

\begin{figure}
    \centering
    \includegraphics[page=1]{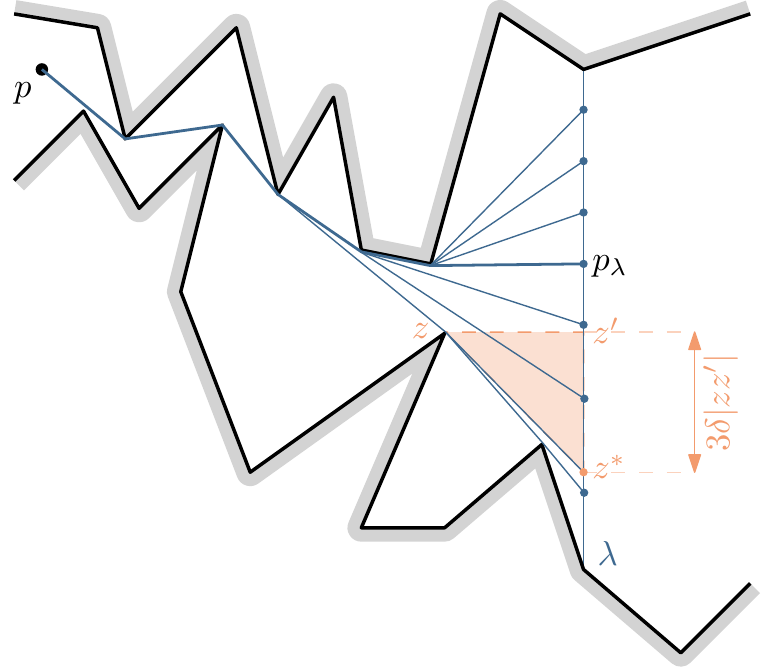}
    \caption{The site $p$ and the additional points $p\ilambda$ it generates for $i \in \{-3,\dots,3\}$.}
    \label{fig:additional_sites}
\end{figure}

In this section, we discuss how to improve the spanning ratio of the low complexity spanner in a simple polygon from $2\sqrt{2} k$ to $2k + \varepsilon$, for any constant $\eps \in (0,2k)$. Observe that this is indeed improves the spanning ratio when $\eps \in (0, (2\sqrt{2}-2)k)$. To apply the refinement of Lemma~\ref{lem:refinement} in combination with our low complexity spanner we require some additional properties on the collection of points added in the refinement. If we apply the lemma directly, and thus add $O(1/\delta^2)$ additional points to the 1-dimensional space for each site in $S$, where $\delta = \eps/(2k)$, then the recursion in the 1-dimensional spanner no longer corresponds to a horizontal split in the polygon. In other words, the proof of Lemma~\ref{lem:group_properties} no longer holds. To fix this issue we make three changes to this naive application. First, we use the paths as defined in Definition~\ref{def:pi_lambda} instead of shortest paths as edges in our spanner. This allows us to construct a low complexity spanner on any tree, instead of requiring the tree to be the shortest path tree of $\lambda$. Second, we do not choose the additional points on $\lambda$ uniformly, but in such a way that the angle of the final segment on the path with $\lambda$ is in a fixed set of angles (of size dependent on $\delta$). Third, instead of constructing a $2k$-spanner on a single 1-dimensional space, we consider a constant number (dependent on $\delta$) of subsets of points, and construct a spanner on each set separately. Next, we elaborate on the second and third change.

\subparagraph{Selecting additional points.}
For a site $p \in S$, consider the closest point $p_\lambda$ on $\lambda$ as before. Let $x$ be a point on $\lambda$. As $P$ is a simple polygon, the absolute value of the angle of the final segment on the path $\pi(p,x)$ (with respect to the $x$-axis) strictly increases as $x$ moves away from $p_\lambda$. This implies that for any $\alpha \in (-\pi,\pi)$ there is at most one point $p_\alpha$ on $\lambda$ such that the final segment of $\pi(p, p_\alpha)$ has angle $\alpha$. Note that for $\alpha = 0$ the point $p_\alpha$ is equal to $p_\lambda$ if $p_\lambda$ is not one of the endpoints of $\lambda$. We use this property to select additional points corresponding to the same angles for each site.

As in Lemma~\ref{lem:refinement}, we select $O(n/\delta^2)$ points on $\lambda$ in total, where $\delta = \eps/(2k)$. Instead of adding these points directly to $S_\lambda$, we create $O(1/\delta^2)$ sets $S\ilambda$ of size $O(n)$. The new points are again weighted by the distance to their original site. We define  $S_\lambda^{(0)} := \{p_\lambda : p \in S\}$. Next, we define the other sets~$S\ilambda$.

For ease of description we do not consider the angle of the final
segment in radians, but measure the angle differently. Let~$z$ be a
vertex of $P$ and let~$z'$ be the horizontal projection of~$z$ on
the line through $\lambda$ (observe that~$z'$ might not be on
$\lambda$). Furthermore, let $z^*$ be a point on $\lambda$ visible
from~$z$. See Figure~\ref{fig:additional_sites}. We then describe
the angle of the segment $zz^*$ by the ratio between $|zz'|$ and
$|z'z^*|$. Let $p\ilambda$ be the point on $\lambda$
for which the last segment on $\pi(p, p\ilambda)$ has ratio
$1/(\delta|i|)$
and lies above $p_\lambda$ if $i > 0$ or below $p_\lambda$ if $i < 0$, if it exists. The set $S\ilambda$ is then defined as $S\ilambda := \{ p\ilambda : p \in S\}$. We will only consider the sets $S\ilambda$ for $i \in I$, where $I = \{-\lceil 3/\delta^2 \rceil, -\lceil 3/\delta^2 \rceil + 1, \dots, \lceil 3/\delta^2 \rceil -1, \lceil 3/\delta^2 \rceil \}$. Let $S_\lambda$ be the union of these sets~$S\ilambda$.

To compute the sets $S\ilambda$ we use the approach from
Lemma~\ref{lem:find_projections} in both $P_\ell$ and $P_r$, which
computes the set $S_\lambda^{(0)}$. For each set $S\ilambda$, instead
of using the horizontal decomposition, we use a decomposition that
consists of trapezoids whose sides are parallel to the angle of the
last segment of the paths in the set $S\ilambda$. Such a decomposition can for example be computed simply by rotating $P$ and $S$ accordingly, and then computing a horizontal decomposition. We then use the coloring of Lemma~\ref{lem:find_projections} to find the site $p\ilambda$ and its distance to $p$ for each $p \in S$. These additional points and the trees that they form can thus be computed in $O(1/\delta^2 \cdot (m + n \log m))$ time.

\subparagraph{Constructing a spanner.} 
To obtain a $(2k + \eps)$-spanner in $P$, we need that for two sites $p \in S_\ell$ and $q \in S_r$ there are points $p_\lambda^{(i)}, q_\lambda^{(j)} \in S_\lambda$ such that
$d_{\G_\lambda}(p_\lambda^{(i)}, q_\lambda^{(j)}) \leq 2k\cdot d_w(p_\lambda^{(i)}, q_\lambda^{(j)})$ and that $d_w(p_\lambda^{(i)}, q_\lambda^{(j)}) \leq  (1+\delta) \cdot d(p,q)$.
We do not require these properties for all pairs of sites in $S$, but only for pairs that lie on opposite sides of $\lambda$. This means that we do not need to build a spanner on the entire set $S_\lambda$. Let $S_{\lambda, \ell}^{(i)}$ and $S_{\lambda, r}^{(i)}$ denote the subsets of $S\ilambda$ that contain all points generated by sites in $S_\ell$ and $S_r$ respectively. We build a 1-dimensional $2k$-spanner $\G_{i,j}$ for each set $S_{i,j} := S_{\lambda, \ell}^{(i)} \cup S_{\lambda, r}^{(j)}$, $i,j \in I$. We thus construct $O(1/\delta^4)$ spanners on sets of size $O(n)$. The edges in $\G_\lambda$ are then simply the union of the edges in all spanners~$\G_{i,j}$. 

To construct a spanner $\G_{i,j}$, we apply our algorithm to construct a low complexity 1-dimensional spanner on a tree $T_{i,j}$ instead of $\mathit{SPT}_\lambda$. The tree $T_{i,j}$ consists of $\lambda$ together with the paths from each point $p\ilambda \in S_{i,j}$ (or $p_\lambda^{(j)}$) to its original site $p$. The following lemma implies that this indeed results in a tree.

\begin{figure}
    \centering
    \includegraphics[page=4]{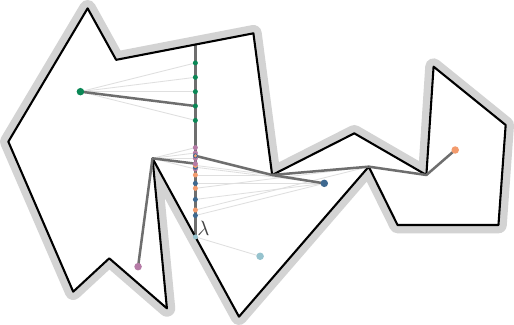}
    \caption{For each site the path to its projection and the
      (valid) additional points for $i \in \{-2, -1, 0, 1, 2\}$ are
      shown. The fat grey tree is the concatenation of the $i = -1$
      paths in $P_\ell$, the $i = 2$ paths in $P_r$, and $\lambda$.}
    \label{fig:tree_selection}
\end{figure}

\begin{lemma}
    For two sites $p,q \in S$ the paths $\pi(p, p\ilambda)$ and $\pi(q, q\ilambda)$ do not properly intersect.
\end{lemma}
\begin{proof}
    For contradiction, let $x$ be the first point on $\pi(p, p\ilambda) \cap \pi(q, q\ilambda)$ as seen from $p$. We know that $x$ defines a unique point $x\ilambda$ on $\lambda$. As $x$ is on $\pi(p, p\ilambda)$ the subpath from $x$ to $p\ilambda$ is also a shortest path. It follows that $x\ilambda = p\ilambda$. Similarly, $x\ilambda = q\ilambda$. We conclude that $p\ilambda = q\ilambda$, and thus the paths from $p$ and $q$ overlap from $x$ on.
\end{proof}

As in the polygonal domain setting, it follows that the complexity of the edges of $\G_{i,j}$ in $P$ is $O(m\sqrt{n} + n\log^2n)$. The complexity of all edges in $\G_\lambda$ in $P$ is thus $O(1/\delta^4 \cdot (m\sqrt{n} + n\log^2n))$.

What remains is to prove this indeed results in the desired spanning ratio for sites on opposing sides of $\lambda$. Without loss of generality, let $p \in S_\ell$ and $q \in S_r$. We want to show that there are points $p_\lambda^{(i)}, q_\lambda^{(j)} \in S_\lambda$ such that
$d_w(p_\lambda^{(i)}, q_\lambda^{(j)}) \leq  (1+\delta) \cdot d(p,q)$. Let $r$ be the intersection point of $\pi(p,q)$ and $\lambda$. The following lemma states that there is a short path in $\G_\lambda$ between $p$ and $r$.

\begin{figure}
    \centering
    \includegraphics[page=3]{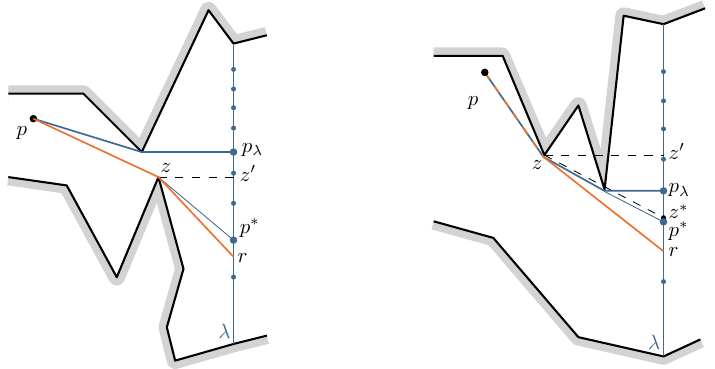}
    \caption{Notation used in Lemma~\ref{lem:1dim_path}.}
    \label{fig:notation_eps}
\end{figure}

\begin{lemma}\label{lem:1dim_path}
    Let $r \in \lambda$ and $p \in S$. Then there is a point $p\ilambda \in S_\lambda$ such that $d(p, p\ilambda) + d(p\ilambda, r) \leq (1+\delta) d(p,r)$.
\end{lemma}
\begin{proof}
    Without loss of generality, we assume that $r$ is below $p_\lambda$ on $\lambda$. Let $p^* \in S_\lambda$ be the point generated by~$p$ on $\pi(r,p_\lambda)$ that is closest to~$r$. Furthermore, let $z$ be the last vertex on $\pi(p,r) \cap \pi(p,p^*)$, and let $z'$ be the horizontal projection of~$z$ on the line through~$\lambda$. See Figure~\ref{fig:notation_eps} for an illustration of the notation. Observe that when $z$ is not the last vertex on $\pi(p,p^*)$ then the last vertex on this path must be above $\pi(p,r)$. We consider the following two cases: either $d(r,p^*) > \delta |zz'|$, or $d(r,p^*) \leq \delta |zz'|$.
    
    If $d(r,p^*) > \delta |zz'|$, we claim that $p^* = p_\lambda^{(i')}$ where $i'= -\lceil 3/\delta^2 \rceil$, i.e. $p^*$ is the outermost point generated by $p$. Consider the point $z^*$ at distance $|zz'| \cdot \delta|i'|$ below $z'$. If $p^* = z^*$, then $d(r,p^*) > \delta |zz'|$ implies that the point $p_\lambda^{(i'-1)}$ should lie between $p^*$ and $r$, which is a contradiction. If $p^* \neq z^*$, then $p^*$ must lie below $z^*$, as it must be generated by a vertex below the segment $zz^*$. It follows that  $d(r,z^*) > d(r,p^*) > \delta |zz'|$. This again implies the contradiction that $p_\lambda^{(i'-1)}$ is between $p^*$ and $r$. We conclude that the claim holds. It follows that $|z'r| \geq |z'z^*| = |zz'| \cdot \delta \lceil 3/\delta^2 \rceil \geq |zz'| \cdot (1/\delta + 2/\delta) \geq |zz'| \cdot (1 + 2/\delta) $, as $\delta \in (0,1)$. Rewriting this equation gives $|zz'| \leq \delta/2 \cdot (|z'r| - |zz'|)$. It then follows that
    \begin{align*}
        d(p, p^*) + d(p^*,r) &= d(p,z) + d(z,p^*) + d(p^*,r)\\
        &\leq d(p,z) + |zz'| + |z'p^*| + d(p^*,r)\\
        &= d(p,z) + |zz'| + |z'r|\\
        &\leq d(p,z) + |zz'| + |z'z| + |zr|\\
        &= d(p,z) + 2|zz'| + |zr|\\
        &\leq d(p,z) + \delta \cdot (|z'r| - |zz'|) + |zr|\\
        &\leq d(p,z) + \delta \cdot (|zz'| + |zr| - |zz'|) + |zr|\\
        &\leq d(p,z) + (1 + \delta)|zr|\\
        &\leq (1 + \delta)d(p,r).
    \end{align*}
    
    If $d(r,p^*) \leq \delta |zz'|$, then
    \begin{equation*}
        d(p, p^*) + d(p^*, r) \leq d(p,z) + d(z,p^*) + \delta |zz'| \leq d(p,z) + d(z,r) + \delta d(z,r) \leq (1+\delta)d(p,r).
    \end{equation*}
    Here, we use that $d(z,p^*) \leq d(z,r)$, as $p^*$ is closer to $p_\lambda$ than $r$, and that $|zz'| \leq d(z,r)$.
\end{proof}

Lemma~\ref{lem:1dim_path} also implies that there is a point $q_\lambda^{(j)} \in S_\lambda$ such that $d(q, q_\lambda^{(j)}) + d(q_\lambda^{(j)}, r) \leq (1+\delta) d(q,r)$. It follows that $d_w(p_\lambda^{(i)}, q_\lambda^{(j)}) \leq d(p, p\ilambda) + d(p\ilambda, r) + d(q, q_\lambda^{(j)}) + d(q_\lambda^{(j)}, r) \leq (1+\delta)(d(p,r) + d(q,r)) = (1 + \delta)d(p,r)$. As $\G_{i,j}$ is a $2k$-spanner on $S_{i,j}$, this results in a spanning ratio of $2k(1+\delta) = 2k +\eps$ for $\G$ when choosing $\delta = \eps/(2k)$.

Combining this result with the spanner construction in Section~\ref{sub:construction_algo_SP}, we obtain the following theorem.

\begin{theorem}\label{thm:2k-eps-spanner-time}
Let $S$ be a set of $n$ point sites in a simple polygon $P$ with $m$ vertices, and let $k \geq 1$ be any integer constant. For any constant $\varepsilon \in (0,2k)$, we can build a relaxed geodesic $(2k + \varepsilon)$-spanner of size~$O(c_{\varepsilon, k}n\log^2n)$ and complexity $O(c_{\varepsilon, k} (m n^{1/k} + n\log^2 n))$ in $O(c_{\varepsilon, k}n\log^2n + m\log n + K)$ time, where $c_{\varepsilon, k}$ is a constant depending only on $\varepsilon$ and $k$, and~$K$ is the output complexity. 
\end{theorem}